\documentclass[journal]{IEEEtran}
\usepackage{amsmath,amsfonts}
\usepackage{algorithm}
\usepackage{array}
\usepackage[caption=false,font=normalsize,labelfont=sf,textfont=sf]{subfig}
\usepackage{textcomp}
\usepackage{stfloats}
\usepackage{url}
\usepackage{verbatim}
\usepackage{graphicx}
\usepackage{cite}
\usepackage{mathtools}
\usepackage{amssymb}
\usepackage{amsthm}
\usepackage{algpseudocode}
\usepackage{flushend}
\theoremstyle{definition}
\usepackage{color}
\newtheorem{theorem}{Theorem}
\theoremstyle{remark}
\newtheorem{lemma}{Lemma}

\newcommand{\ie}{\emph{i.e.}}

\begin{document}

\newcommand{\re}{\operatorname{Re}}

\title{Optimum Beamforming and Grating Lobe Mitigation for Intelligent Reflecting Surfaces}
\author{Sai Sanjay Narayanan, Uday~K~Khankhoje,~\IEEEmembership{Senior~Member,~IEEE},  Radha Krishna Ganti~\IEEEmembership{Member,~IEEE}%
\thanks{\textit{Corresponding author: Uday Khankhoje (e-mail: uday@ee.iitm.ac.in)}.}%
\thanks{The authors are with the Department of Electrical Engineering, Indian Institute of Technology Madras, Chennai, India}%
\thanks{Authors acknowledge the following research grants: (i)``6G: Sub-THz Wireless Communication with Intelligent Reflecting Surfaces (IRS)," R-23011/3/2022-CC\&BT-MeitY by the Ministry of Electronics and Information Technology (MeitY), Government of India, and (ii) ``Intelligent approaches to designing and building Intelligent Reflecting Surfaces for 6G applications," by the Qualcomm 6G University Research India Program.}}
\maketitle

\begin{abstract}
Ensuring adequate wireless coverage in upcoming communication technologies such as 6G is expected to be challenging. This is because user demands of higher datarate require an increase in carrier frequencies, which in turn reduce the diffraction effects (and hence coverage) in complex multipath environments. Intelligent reflecting surfaces have been proposed as a way of restoring coverage by adaptively reflecting incoming electromagnetic waves in desired directions. This is accomplished by judiciously adding extra phases at different points on the surface. In practice, these extra phases are only available in discrete quantities due to hardware constraints. Computing these extra phases is computationally challenging when they can only be picked from a discrete distribution, and existing approaches for solving this problem were either heuristic or based on evolutionary algorithms. We solve this problem by proposing fast algorithms with provably optimal solutions. Our algorithms have linear complexity, and are presented with rigorous proofs for their optimality. We show that the proposed algorithms exhibit better performance. We analyze situations when unwanted grating lobes arise in the radiation pattern, and discuss mitigation strategies, such as the use of triangular lattices and prephasing techniques, to eliminate them. We also demonstrate how our algorithms can leverage these techniques to deliver optimum beamforming solutions.
\end{abstract}


\IEEEpeerreviewmaketitle

\section{Introduction}
Intelligent reflecting surface (IRS) is proposed as a key technology to enable the advancement of 6G wireless communications, which demands high data rates and a large spectrum of frequencies to operate \cite{9771079}. IRS is a planar surface consisting of many sub-wavelength passive reflecting elements, where each element is equipped with a controller that decides the phase shift it can impart to the incoming beam. By carefully tuning these phase shifts, one can make the reflected beams from each element constructively interfere in a given direction, and destructively interfere in other directions. On account of being a passive surface, an IRS exhibits energy efficiency and improves the spectral efficiency of the system \cite{7510962}.

The required phase shifts of the elements are obtained by solving an appropriate beamforming optimization problem. Several beamforming problems have been posed in the IRS literature \cite{elbir2023twenty}. Minimum variance beamforming (MVB) involves selecting the weights so as to minimize the variance of the output signal's energy in a given direction while maintaining unity gain in that direction \cite{5447076, 1420809_boyd_mvb, 8721535_vorobyov_mvdr}. In transmit beamforming and broadcast beamforming, the goal is to minimize transmit power while maintaining a signal-to-noise ratio (SNR) threshold for each user \cite{1634819, 6952703, 8811733_irs_enhanced_joint_beamforming}. The above beamforming problems are typically framed as continuous optimization problems, where the weights (and hence the phase shifts) belong to a continuous set of values. However, continuous phase shifters are difficult to implement on the IRS due to the limited size of each element and associated hardware complexity. In practice, each element has a discrete set of phase shifts (e.g., 1-bit or 2-bit), which gives rise to discrete beamforming problems \cite{9110889_hybrid_disc_beamforming, 9133142_channel_estimation, 9305278_iot_beamforming}. While 1-bit phase shifters in particular are easy to implement, the price to pay is the emergence of quantization lobes, which we refer to as the grating lobes induced due to the 1-bit nature of each weight. A common technique to mitigate these lobes (called prephasing) is to perturb the baseline phases of some or all elements by some predetermined strategy \cite{yin_single-beam_2020,kashyap_mitigating_2020,9807634_prephase_1-bit_low-sll,10195171_quantization_suppression,10179250_prephase_low_res_hybrid_phasing}, an idea that has evolved from mitigation strategies in phased antenna arrays \cite{smith1983comparison} and reconfigurable reflectarray antennas \cite{yang2017study}. 

Discrete optimization problems are generally hard to solve efficiently. However, for certain classes of problems there exist techniques that yield approximate solutions. In the context of beamforming, two commonly used methods for solving discrete optimization problems are semi-definite relaxation (SDR), which involves relaxing the problem into a convex optimization problem that can be efficiently solved using interior point methods \cite{zhang_quadratic_2000, ben-tal_lectures_2001}, and quantization, which involves solving the continuous version of the problem and then discretizing the solution. Quantization in particular has been used and studied extensively \cite{4476079_no_of_bits, 5068400_linear_reflectarray, 7389996_hybrid_large_scale, yang2017study, zhang_reconfigurable_2020}. The problem with these methods is that SDR has a high computational complexity, while quantization is only a heuristic method that does not provide any mathematical guarantees of optimality. The beamforming problem is posed as a mixed integer nonlinear program in \cite{wu2019beamforming} and an optimal algorithm is proposed; however, this has exponential worst-case complexity due to which a suboptimal relaxation is proposed via a successive refinement algorithm. Machine learning \cite{yan2022machine} and quantum algorithms \cite{quantum2024} have also been explored for IRS beamforming. A recent review \cite{pan2022overview} comprehensively covers the various strategies employed in the literature for continuous and discrete variable beamforming as well as related issues such as channel estimation, transmission design, and localization.

More recently, a linear complexity algorithm was proposed \cite{zhang2022configuring,ren2022linear} for solving the discrete beamforming problem of maximizing 
the power of the reflected signal in a given direction when a plane wave impinges on the IRS; while optimal, it assumes the IRS elements to have unit reflectivity and the phase shifts to be equiangular. Our work lifts the above assumptions and proposes a series of new algorithms that are both, provably optimal, as well as linear in computational complexity. Building on our previous work \cite{beamforming_conf_2024}, we first present a novel $\mathcal{O}(n)$ complexity algorithm for the 1-bit case where all phase shifts of the individual elements are either $0^\circ$ or $180^\circ$ (relatively). This is already a vast improvement over heuristic techniques or evolutionary algorithms. The 1-bit algorithm is then generalized to the case where each weight is constrained to its own \textit{arbitrary} discrete two-element set, i.e.~the elements of this set need not have unit modulus, nor a phase difference of $180^\circ$. We further generalize the algorithm to the case where the weights are constrained to a $k$-element set ($k > 2$), and the generalization is once again proven to be efficient and optimal. Finally, we propose yet another generalization for the case where beamforming in multiple directions is desired. To the best of our knowledge, ours is the first provably optimal, linear complexity algorithm to solve the beamforming problem under the constraint of arbitrary discrete weights. Using our algorithms to obtain the desired weights, we perform theoretical analyses and simulations for the behavior of sidelobe levels (SLL) for the 1-bit case. We identify the emergence of grating lobes due to a symmetry in the IRS array factor. To this end, we propose two mitigation strategies. In the first approach we propose the use of a triangular lattice to realize the IRS and show supporting numerical results. In the second approach, we demonstrate how our algorithm can be used to incorporate the various grating-lobe mitigation strategies proposed in earlier works.

The paper is structured as follows. In Section II, we formulate the problem statement and present the novel beamforming algorithm along with its generalization to the $k-$element and multiple beams case. In Section III, we present numerical results to establish the superiority of our approach over heuristic approaches. Using our generalized algorithm, we study the trade-off between the number of bits of allowed phase shifts and beamforming error. In Section IV, we present theoretical analyses and detailed results pertaining to sidelobe levels and identify the issue of grating lobes. This is followed by a discussion on grating lobe mitigation strategies and their integration with our algorithm. Finally, we conclude in Section V. 

\section{Optimum beamforming algorithm \& generalizations}
\subsection{Problem Statement}

The IRS is modeled as an $M \times N$ rectangular grid of equi-spaced unit cells in the $x-y$ plane. It is illumined by an incident electromagnetic plane wave whose coordinates according to the forward scatter alignment (FSA) convention \cite{elachi1990radar} are given by $(\theta_{in}, \phi_{in})$, as seen in Fig.~(\ref{fig:IRS_diagram}). As per the FSA convention, the incident beam wavevector is given by $\hat{k}_{in} = \sin{\theta_{in}} \cos{\phi_{in}} \hat{x} + \sin{\theta_{in} \sin{\phi_{in}} \hat{y}} - \cos{\theta_{in}} \hat{z}$, while the reflected wavevector in the direction $(\theta_{0}, \phi_{0})$ is given by $\hat{k}_{0} = \sin{\theta_{0}} \cos{\phi_{0}} \hat{x} + \sin{\theta_{0} \sin{\phi_{0}} \hat{y}} + \cos{\theta_{0}} \hat{z}$ (note the difference in the sign of the $\hat{z}$ component between incident and reflected wavevectors). Each cell imparts a complex valued reflection coefficient (called the cell  \emph{weight}),  $w_{m,n}$, to the incident beam. In order to capture realistic IRS implementations we assume that the weights take on discrete values, and therefore are modeled as belonging to a discrete alphabet, $\mathcal{S}$ (for e.g.,~in many practical implementations, $\mathcal{S} =\{-1,1\}$). 

The normalized array factor of the IRS, $G$, is given by:
\begin{align}
 \label{afac}
    G(\theta, \phi) &= \frac{1}{MN}\sum_{m = 1}^M \sum_{n = 1}^N w_{m,n} e^{j \varphi_{m,n}(\theta, \phi)},
\end{align}
where
\begin{align}
      \label{phi_mn}
 \varphi_{m,n}(\theta, \phi) &= \frac{2 \pi d}{\lambda} \Big( - m \sin{\theta}\cos{\phi} - n \sin{\theta} \sin{\phi} \\   \nonumber
    &{}  + m \sin{\theta_{in}} \cos{\phi_{in}} + n \sin{\theta_{in}} \sin{\phi_{in}} \Big ),
\end{align}
with $d$ as the inter-cell spacing, and $\lambda$ being the wavelength of the incident wave. The array factor gives the far-zone electric field of the array when multiplied by the radiation pattern of the constitutive array cell \cite{balanis2016antenna} (ignoring mutual coupling).

\begin{figure}[htb!]
    \centering
    \includegraphics[width = 0.36 \textwidth]{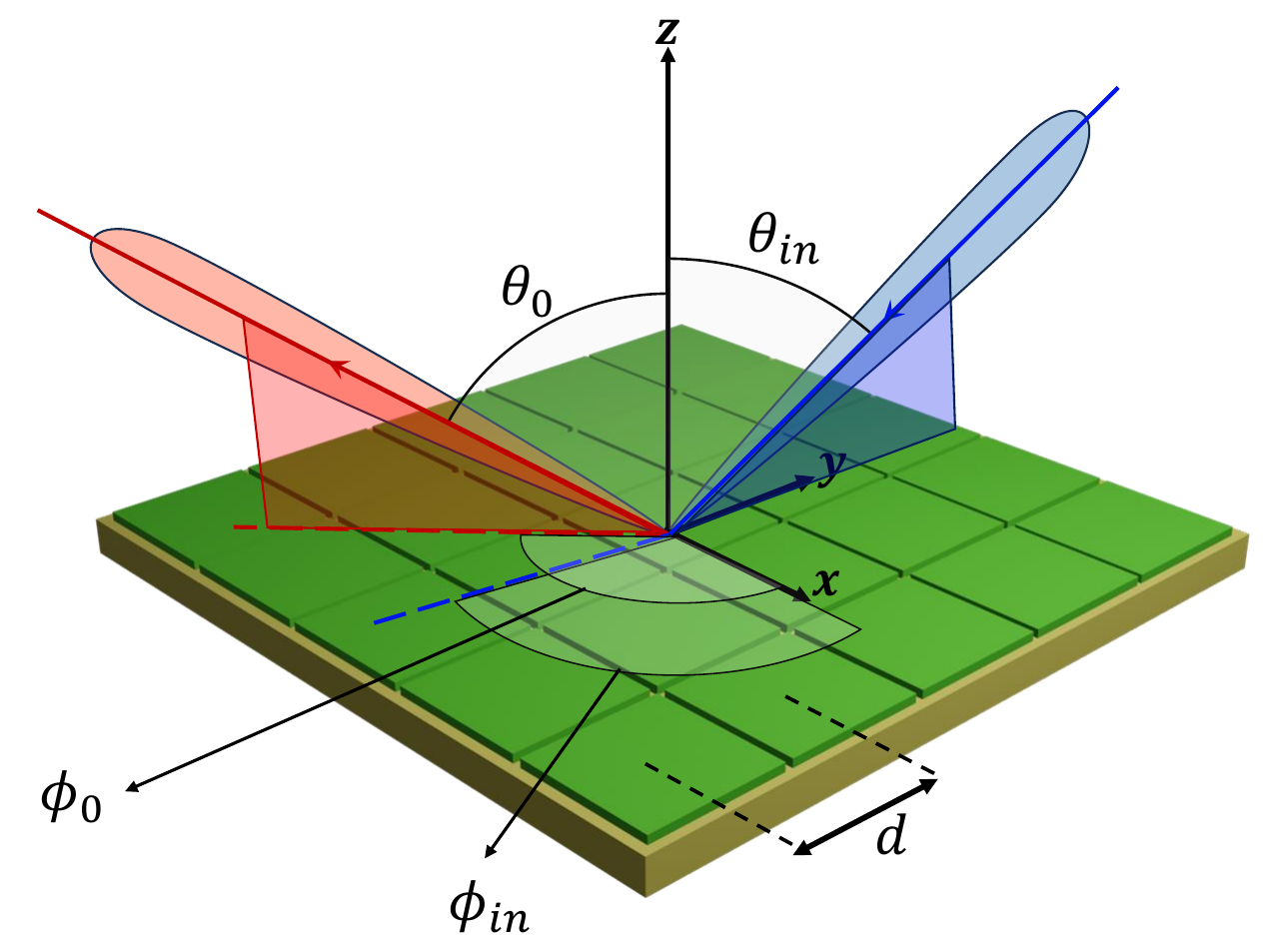}
    \caption{A schematic diagram of a $5\times 5$ IRS, with the incident beam wavevector (FSA coordinates $(\theta_{in}, \phi_{in} $)) represented in blue and the reflected beam wavevector (FSA coordinates ($\theta_0, \phi_0$)) represented in red. }
    \label{fig:IRS_diagram}
\end{figure}

We now formally state the problem of IRS beam forming in a certain direction $(\theta_0, \phi_0)$ as:
\begin{align} \label{disc_opt_1}
    \nonumber \max_{w_{m,n}} \quad &|G(\theta_0, \phi_0)|^2, \\
    \textrm{s.t.} \quad &w_{m,n} \in \mathcal{S},\hspace{4mm} \forall m,n.
\end{align}
This is a discrete optimization problem since the weights $w_{m,n}$ belong to a finite set. We note that the problem of beamforming can take on many more sophisticated forms as compared to \eqref{disc_opt_1} above (see \cite{elbir2023twenty} for an extensive review); in this work we focus on the simpler, single beam, noise-free analysis.

\subsection{Thresholding Method}
Thresholding (also known as quantization) is the most common discrete beamforming algorithm used in the literature. Consider the case where the allowed values for each cell are: $\mathcal{S} = \{1, -1\}$, \ie,~a 1-bit IRS, which corresponds to each unit cell having phase shift of $0^\circ$ or $180^\circ$. 

In the thresholding method, the main idea is to first obtain the optimal solution to the continuous version of the problem (\ie,~solving \eqref{disc_opt_1} with the constraint as $w_{m,n} \in \mathbb{C}$, $\lvert w_{m,n} \rvert = 1$), as done in \cite[Thm.~1.2.5]{bhattacharyya2006phased}. The solution is unique up to multiplication by a unit-magnitude complex number. Once the continuous solution is obtained, we discretize the weights by mapping them to whichever among $\{1,-1\}$ they are closest to. The complete methodology is presented in Algorithm 1.

\begin{algorithm}[htb]
    \caption{Thresholding Algorithm} \label{th_algo}
    \begin{algorithmic}[1]
        \Require $\theta_{in}, \phi_{in}, \theta_{0}, \phi_{0}$
        \For{$1 \leq m \leq M$, $1 \leq n \leq N$}
        \State Compute $w_{m,n} =  \exp{\big(-j \varphi_{m,n} (\theta_0, \phi_0) \big)}$
        \If{$\angle w_{m,n} \in [-\frac{\pi}{2}, \frac{\pi}{2})$}
            \State $\widetilde{w}_{m,n} \gets 1$
        \Else
            \State $\widetilde{w}_{m,n} \gets -1$
        \EndIf
        \EndFor
        \State \Return discrete weights $\widetilde{w}_{m,n}$
    \end{algorithmic}
\end{algorithm}

The thresholding method is popular on account of being intuitive and easy to implement \cite{zhang_reconfigurable_2020,yang2017study}. However, due to its heuristic nature it does not have any guarantees of optimality. Further, it is not clear how to generalize the algorithm if the elements of $\mathcal{S}$ have unequal amplitude. These observations pave the way for our proposed method, which we present next.

\subsection{Optimal Partitioning Algorithm}
The proposed method is termed the Optimal Partitioning Algorithm (OPA), which is based on the following intuition. The sum of complex exponentials in \eqref{afac} can be re-written as $G = \sum_i w_i z_i$, where the index $i$ runs over all tuples $(m,n)$, $w_i$ represents the weight of the $i^{th}$ unit cell, and $z_i = e^{j \varphi_i(\theta, \phi)}$. To build an intuitive understanding of the algorithm about to be presented, we consider the following two examples.

(i) Consider the $z_i$'s as shown in Fig.~\ref{fig:unit_circ}(a). To maximize $|G|$ as per \eqref{disc_opt_1}, the weight $w_i$ associated with each $z_i$ must be chosen appropriately. Since the $w_i$s are restricted to $\{-1,1\}$, it follows that $|G|$  is maximized when all the $z_i$s in the left-half of the plane are assigned $w_i=-1$ (indicated via the red arrows in the Figure) and those in the right-half plane are assigned $w_i=1$ (or vice-versa). In effect, we have found a partitioning line (in this case the $y$-axis) that separates the assignment of the weights. The resulting weight assignment is equivalent to the thresholding algorithm.\\
(ii) Next consider the $z_i$'s as shown in as shown in Fig.~\ref{fig:unit_circ}(b). As can be seen, the $y$-axis no longer partitions the weights for maximizing $|G|$; instead, a different weight assignment gives the optimal sum, as shown by the red arrow in the Figure. Thus, the thresholding algorithm would have given a suboptimal weight assignment. 

We note that in both examples, the optimal assignment can be conceived of as the result of an optimal partitioning line as indicated by the dashed blue lines in Fig.~\ref{fig:unit_circ}. 
\begin{figure}[htbp]
\begin{center}
\includegraphics[width = 0.45\textwidth]{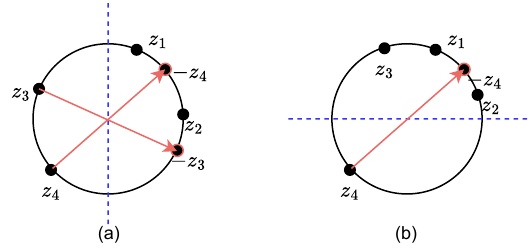}
\caption{Complex argand plane with $z_i$s on the unit circle; in (a) the assignment of weights $-1$ to $z_3,z_4$ maximizes $|G|$, while in (b) the assignment of weight $-1$ to $z_4$ maximizes $|G|$. The dashed lines indicate the optimal partitioning for weight assignments in each case.}
\label{fig:unit_circ}
\end{center}
\end{figure}

The above intuition of an optimal partitioning is formalized via the following theorem.

\begin{theorem}
    Let $z_1, \dots z_n$ be arbitrary non-zero complex numbers, and let $\widetilde{w}_1, \dots \widetilde{w}_n$ be an optimal solution to the optimization problem
    \begin{align} \label{thm1}
        \max_{w_i} \quad &\big| w_1 z_1 + \dots + w_n z_n \big|, \\ \nonumber
        \textrm{s.t.} \quad &w_i \in \{1,-1\}, \hspace{3mm} \forall i = 1, \dots n.
    \end{align}
    Then, there exists a line $\ell$ passing through the origin with the following properties:
    \begin{enumerate}
        \item $\ell$ does not contain any $z_i$.
        \item $\widetilde{w}_i = 1$ for all $z_i$ lying on one side of the line, and $\widetilde{w}_i = -1$ for all $z_i$ on the other side. 
    \end{enumerate}
\end{theorem}
\begin{proof}
    The proof is presented in Appendix \ref{theorem 1 proof}.
\end{proof}

The above theorem asserts that \emph{if} an optimal set of weights exists, then there exists a \emph{separating line} passing through the origin, which partitions the set $\{z_1, \dots z_n \}$ into two sets $\mathcal{A}$ and $\mathcal{B}$, where $\widetilde{w}_i = 1$ for each $z_i \in \mathcal{A}$, and $\widetilde{w}_i = -1$ for each $z_i \in \mathcal{B}$. However, we know that an optimal solution exists because the feasible set is finite. Therefore, an optimal separating line definitely exists, and our goal should be to find one such line. Once the partition $\{ \mathcal{A}, \mathcal{B}\}$ is obtained, the assignment of $1$ and $-1$ to each set is arbitrary, as replacing $\widetilde{w}_i$ with $-\widetilde{w}_i$ for each $i = 1, \dots n$ will yield the same value of $\big| \sum_i w_i z_i \big|$. 

The natural question that arises is: how many partitions should we check? We must check all partitions that arise due to separating lines passing through the origin. However, multiple separating lines can lead to the same partition, as shown in Fig.~\ref{fig:separating_line}. In particular, one has the freedom to rotate a separating line anticlockwise about the origin without changing the partition, until it hits some $z_i$. Therefore, the partition obtained by \emph{any} separating line is equivalent to one obtained by a separating line whose angle with the $x$ axis is infinitesimally lesser than the argument of some $z_i$, and it can be easily seen that one of the subsets of the partition must be of the form $ \{ z_j \mid \arg z_i \leq \arg z_j < \arg z_i + \pi \}$ for some $i$ (see Fig.~\ref{fig:separating_line}).
Since there are $n$ complex numbers $z_i$, only $n$ such separating lines (and therefore only $n$ such partitions) need to be investigated. So, our optimal partitioning seeking algorithm has linear time complexity.
\begin{figure}[htb] 
    \centering
    \includegraphics[width =0.23\textwidth]{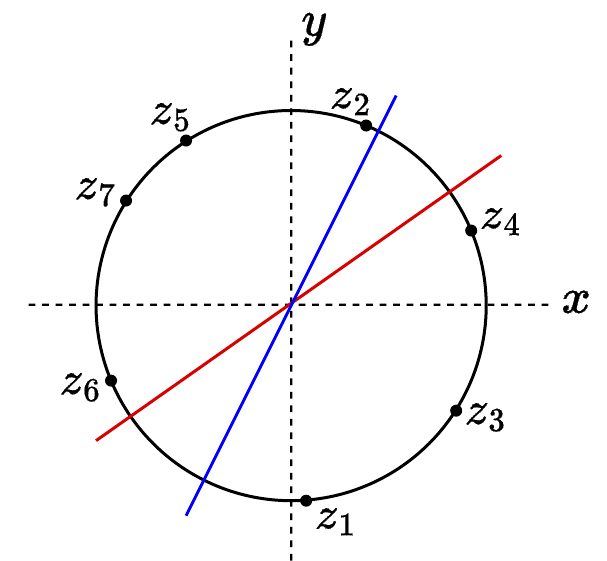}
    \caption{An example where $n = 7$ and $z_1, \dots z_7$ all lie on the unit circle. The red and blue separating lines both give rise to the same partitions, $\mathcal{A} = \{ z_1, z_3, z_4\}$ and $\mathcal{B} = \{ z_2, z_5, z_6, z_7 \}$. Further, the set $\{z_j \big| \arg{z_2} \leq \arg{z_j} < \arg{z_2} + \pi \}$ is precisely $\mathcal{B}$. The blue line can be thought of as having angle infinitesimally lesser than $\arg{z_2}$.}
    \label{fig:separating_line}
\end{figure}

Theorem 1 can now be used to formulate the Optimal Partitioning Algorithm (OPA), which computes the optimal weights for the optimization problem in \eqref{disc_opt_1} with $\mathcal{S}=\{-1,1\}$. The algorithm is as follows: first, compute the $\{ z_i \}$ using \eqref{phi_mn}. Then, choose any arbitrary line passing through the origin, and obtain all possible partitions of $\{z_i\}$ by rotating the line anticlockwise about the origin. When rotating the line, the partition can only change when the line crosses some $z_i$, and hence we will obtain at most $n = MN$ unique partitions. Next, for each partition, assign $w_i = 1$ for all $z_i$ on one side of the partition, and $w_i = -1$ for all $z_i$ on the other side. Compute $|\sum_i w_i z_i |$ for each partition, and choose the partition that corresponds to the largest value of $|\sum_i w_i z_i |$. Theorem 1 guarantees that the weight assignment corresponding to the resulting partition is optimal.

The complete method is displayed in Algorithm 2; since we only need to check $n = MN$ partitions, OPA will find the optimal weights in time $\mathcal{O}(MN)$. The optimality of the obtained solution is preserved up to a change of sign, \ie,~if $\{\widetilde{w}_i\}$ are optimal, then so are $\{-\widetilde{w}_i\}$, since both sets give the same value of $|G|$.

\begin{algorithm}[htb!]
    \caption{Optimal Partitioning Algorithm (OPA)} \label{non-sdp-1,-1}
    \begin{algorithmic}[1]
        \Require Non-zero complex numbers $z_1, \dots z_n$ 
        \State max $\gets 0$ 
        \For{$i = 1$ to $n$}
        {
            \State Compute $\delta_i = \arg{z_i}$. 
            \State Set $w_j = 1$ for all $j$ such that $\arg{z_j} \in [\delta_i, \delta_i + \pi)$ and $w_j = -1$ for others (\textit{corresponds to assigning opposite weights to either side of the partition})
            \State Compute $g = \big| w_1 z_1 + \dots w_N z_N \big|$.
            \If{$g > $  max}
                \State max $\gets g$
                \State optimal weights $\gets \begin{bmatrix}
                    w_1 & \dots & w_n
                \end{bmatrix}$.
            \EndIf
        }
        \EndFor
        \State \Return optimal weights
    \end{algorithmic}
\end{algorithm}

We now discuss a generalization of the above algorithm for the case where each weight is constrained to its own binary set, \ie, $w_i \in \{ a_i, b_i\}$ for all $i = 1, \dots n$. In order to solve the resulting optimization problem, we introduce the new scaled variables $y_i$, such that $ w_i = \big( \frac{a_i + b_i}{2} \big) + y_i \big( \frac{a_i - b_i}{2} \big)$, and rewrite the sum in terms of the variables $y_i$. Since the constraints $w_i \in \{ a_i, b_i \}$ are mapped to $y_i \in \{1,-1\}$, OPA can be used to obtain the optimal solution $\widetilde{y}_i$, which in turn is used to compute the optimal weights $\widetilde{w}_i$. This generalization is termed as Generalized OPA (gOPA), and the full details are displayed in Algorithm~\ref{non-SDP ai bi}.

\begin{algorithm}[htb]
    
    \caption{Generalized OPA (gOPA)} \label{non-SDP ai bi}
    \begin{algorithmic}[1]
        \Require Non-zero complex numbers $z_1, \dots z_n$, and the constraint sets $\{a_1, b_1\}, \dots \{a_n, b_n\}$ for each of the weights (where $a_i \neq b_i$ for all $i$). 
        \vspace{2mm}
        \State Compute $z_1', \dots z_n', z_{n+1}'$ as:
        \begin{align*}
             &z_i' = \Big( \frac{a_i - b_i}{2} \Big)z_i \hspace{4mm} \forall i = 1, \dots n, \\
             &z_{n+1}' = \sum_{i=1}^n \Big( \frac{a_i + b_i}{2} \Big)z_i.
        \end{align*}
        \State Apply OPA to the set $z_1', \dots z_{n+1}'$, and get optimal weights $\begin{bmatrix}
            \widetilde{y}_1 & \dots & \widetilde{y}_n & \widetilde{y}_{n+1}
        \end{bmatrix}$.
        \If{$\widetilde{y}_{n+1} = -1$}
                \For{$i = 1$ to $n$}
                    \State $\widetilde{y}_i \gets - \widetilde{y}_i$.
                \EndFor
        \EndIf
        \State Compute the optimal weights $\widetilde{w}_1, \dots \widetilde{w}_n$ as $$\widetilde{w}_i = \Big( \frac{a_i + b_i}{2} \Big) + \widetilde{y}_i \Big( \frac{a_i - b_i}{2} \Big).$$
        \State \Return optimal weights $\begin{bmatrix}
            \widetilde{w}_1 & \dots & \widetilde{w}_n
        \end{bmatrix}$.
    \end{algorithmic}
\end{algorithm}

\subsection{$k-$element Generalization}
We now present yet another generalization of OPA called kOPA, where each weight is constrained to a $k-$element set with $k > 2$, \ie,~$w_{m,n} \in \mathcal{S} = \{a_1, \dots a_k \}$ for all $m,n$, and $a_1, \dots a_k$ are arbitrary complex numbers. 

As in the previous case, it turns out that the optimal solution is once again obtained by partitioning the $z_i$'s into $k$ subsets with a specific geometric structure, where each $z_i$ in a given subset gets the same weight. In OPA, we saw that the optimal partition arises from considering a separating line, which can be thought of as two diametrically opposite rays emerging from the origin. A similar intuition carries forward to kOPA, where the optimal partition arises due to a configuration of $k$ such rays. 

Let $\mathcal{S} = \{a_1, \dots a_k \}$ be the discrete weight set from which the weights $w_{m,n}$ have to be picked. We then define the following regions:
\begin{align}
    \label{A_i}
    A_i = \{ z \in \mathbb{C} \mid \operatorname{Re}( (a_i - a_j)z ) > 0 ,\forall j \neq i \}.
\end{align}
For a given $j$, the inequality $\operatorname{Re}( (a_i - a_j)z ) > 0 $ represents a half-space in $\mathbb{C}$ (or equivalently $\mathbb{R}^2$), and hence $A_i$ is an intersection of $k-1$ half-spaces. Therefore, the collection $ \{A_1, \dots A_k \}$ can be easily seen as a radial partition $\mathcal{P}$ of the complex plane with  $A_i$'s as the cones of the radial partition. Let  $\theta \in [0, 2\pi]$, and let $\mathcal{P}(\theta)$ denote the rotation of the partition $\mathcal{P} $ by an angle $\theta$ anticlockwise, with the corresponding rotated cones denoted by $A_i(\theta)$. An illustration of $\mathcal{P}$ and $\mathcal{P}(\theta)$ is provided in Fig~\ref{fig:kopa_partition_and_rotation}.

\begin{figure}[htb!]
    \centering
    \includegraphics[width = 0.5 \textwidth]{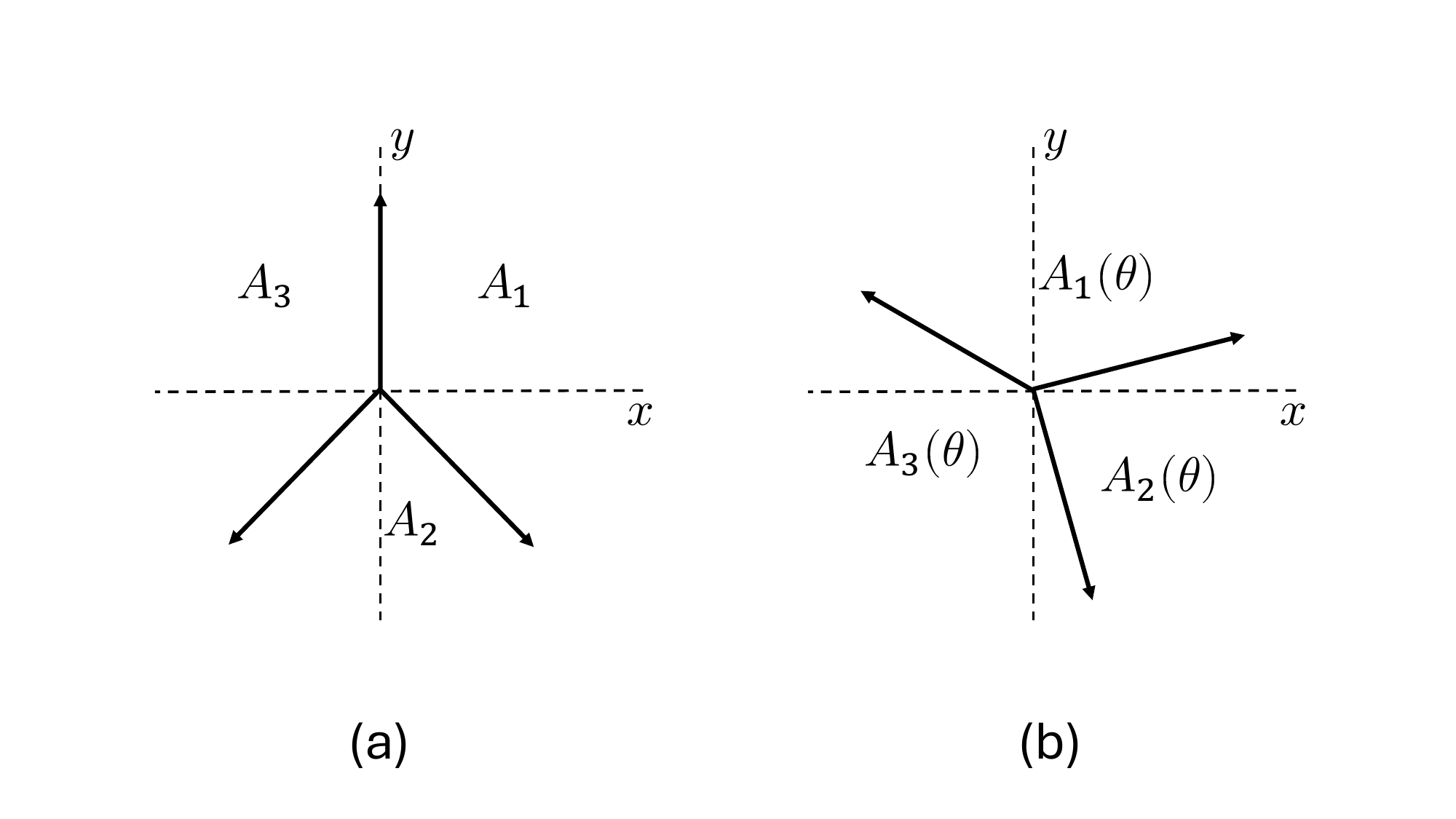}
    \caption{(a) Depiction of $\mathcal{P} = \{A_1, A_2, A_3 \}$ for the constraint set $\mathcal{S} = \{1, j, -1\}$ (b) Rotated partition $\mathcal{P}(\theta) = \{ A_1(\theta), A_2(\theta), A_3(\theta) \}$, where $\theta = 60^\circ$.}
    \label{fig:kopa_partition_and_rotation}
\end{figure}

The idea behind kOPA is that the optimal partition is obtained by a suitable rotation of the partition $\mathcal{P}$, as described in the following theorem.

\begin{theorem}
    Let $z_1, \dots z_n$ be arbitrary non-zero complex numbers, and let $\widetilde{w}_1, \dots \widetilde{w}_n$ be an optimal solution to the optimization problem
    \begin{align} \label{thm2}
        \max_{w_i} \quad &\big| w_1 z_1 + \dots + w_n z_n \big|, \\ 
        \nonumber \textrm{s.t.} \quad &w_i \in \mathcal{S} = \{a_1, \hdots, a_k \}, \hspace{3mm} i = 1, \hdots, n.
    \end{align}
    Then, there exists $\delta \in [0,2\pi]$ such that
    \begin{enumerate}
        \item None of the $z_i$ lie on the edges of the cones of the rotated partition  $\mathcal{P}(\delta)$.
        \item  If  $z_i$  belongs to the cone $A_j(\delta)$ in the rotated partition $\mathcal{P}(\delta)$, then     $\widetilde{w}_i = a_{j}$.
    \end{enumerate}
\end{theorem}

\begin{proof}
    The proof is presented in Appendix \ref{theorem_2_proof}.
\end{proof}

We can now formulate kOPA as follows. Firstly, compute $\mathcal{P}$. Then, rotate $\mathcal{P}$ by $\delta$ (sweeping from $0$ to $360^\circ$), to obtain various partitions. For each of these partitions, assign weights according to Theorem 2, and compute $|\sum w_i z_i|$. Finally, choose the partition that corresponds to the largest value of $|\sum w_i z_i|$. The details are formalized in Algorithm \ref{kOPA_new}. Theorem 2 guarantees that the chosen weight assignment is optimal. An example of optimal weight assignment is shown in Fig~\ref{fig:kopa_example}.

\begin{algorithm}[htb]
    
    \caption{k-Optimal Partitioning Algorithm (kOPA)} \label{kOPA_new}
    \begin{algorithmic}[1]
        \Require Non-zero complex numbers $z_1, \hdots, z_n$, and the constraint set $\mathcal{S} = \{a_1, a_2, \dots a_k\}$. 
        \vspace{2mm}
        \State Compute the radial partition $\mathcal{P} = \{A_1, \hdots, A_k\}$ as described.
        \State max $\gets 0, \delta \gets 0$.
        \While{$\delta \leq 2 \pi$}
            \State Find the smallest $\psi \in (0, 2\pi]$ such that some $z_i$ lies on one of the edges of the cones of $\mathcal{P}(\delta + \psi)$.
            
            \State Rotate the current partition $\mathcal{P}(\delta)$ by some $\psi' \in (0, \psi)$ to obtain $\mathcal{P}(\delta + \psi')$.
            
            \State Assign $w_i = a_j$ whenever $z_i \in A_j(\delta + \psi')$, for all $i$.
            
            \State Compute $g = \big| w_1 z_1 + \hdots + w_N z_N \big|$.
            
            \If{$g > $  max}
                \State max $\gets g$, optimal weights $\gets \begin{bmatrix}
                    w_1 & \dots & w_n
                \end{bmatrix}$.
                \EndIf
            \State $\delta \gets \delta + \psi$.
        \EndWhile
        \State \Return optimal weights
    \end{algorithmic}
\end{algorithm}

\begin{figure}[htb!]
    \centering
    \includegraphics[width = 0.22 \textwidth]{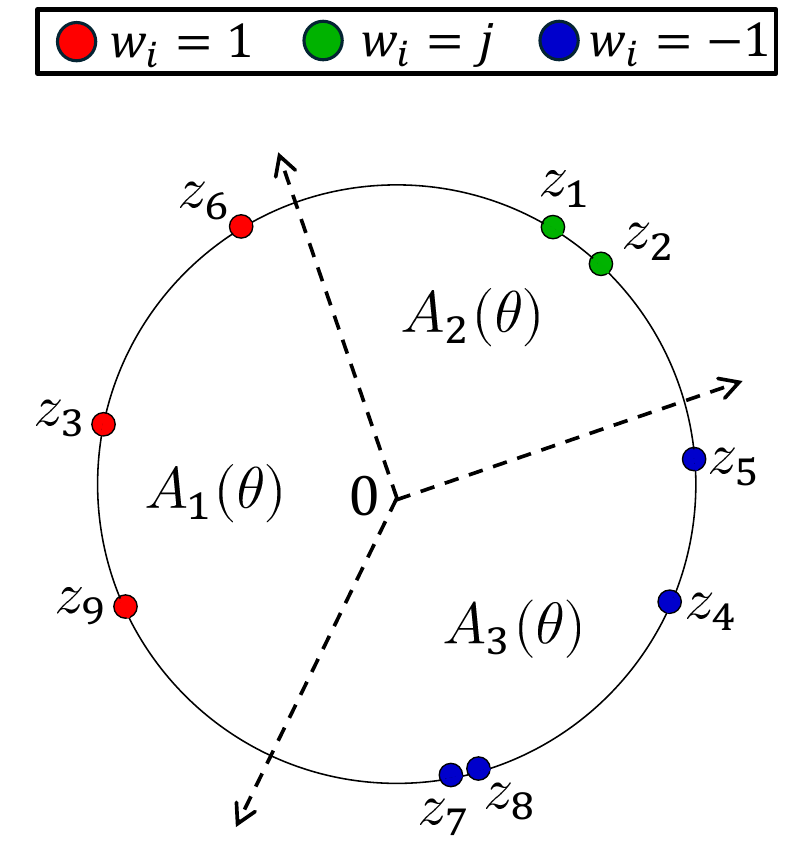}
    \caption{The optimal weight assignment for a set of 9 randomly generated unit modulus complex numbers $z_1, \dots z_9$ (with $ \mathcal{S} = \{1, j, -1\})$, which can be seen as arising from a rotated partition $\mathcal{P}(\theta)$.}
    \label{fig:kopa_example}
\end{figure}
We make the following observations.
\begin{itemize}
    \item The number of unique partitions obtained in kOPA is at most $nk$. This is because when rotating $\mathcal{P}$, the partition only changes when one of the $k$ edges crosses some $z_i$, and this can happen in at most $nk$ unique ways. As a consequence, kOPA has time complexity $O(nk)$ (where $n = MN$ in the context of IRS beamforming), making it extremely fast.
    \item In practical implementations, $k$ is usually equal to $2^b$ where $b$ denotes the number of bits, and the allowed phase shifts are equispaced by $\frac{2 \pi }{k}$, which corresponds to $\mathcal{S} = \{1, e^{j \frac{2 \pi}{k}}, \dots e^{j \frac{2 \pi(k-1)}{k}} \}$. In this case, it can be seen that the cones $A_i, \hdots, A_k$, each subtend an angle of $\frac{2 \pi}{k}$ at the origin. 
\end{itemize}

\subsection{Beamforming in multiple directions}
We now discuss the problem of forming simultaneous beams in multiple directions. Suppose $(\theta_1, \phi_1), \hdots, (\theta_l, \phi_l)$ are the directions in which we want to beamform, and denote the array factor along these directions as $G_1, \hdots, G_l$, where $G_j = G( \theta_j, \phi_j)$ for all $j \in \{1, \hdots, l\}$. Note that we can write $G_j = \sum_{i} w_{i} z_{i}^{(j)}$ where $i$ runs over the tuples $(m,n)$, and hence each $G_j$ is a function of the weights $\{w_i\}$. The optimization problem that we wish to solve is 
\begin{align} \label{multibeam_opt}
    \nonumber \max_{w_i} \quad &|G_1| + |G_2| + \hdots + |G_l|, \\
    \textrm{s.t.} \quad &w_{i} \in \mathcal{S},\hspace{4mm} \forall i.
\end{align}
We now construct an equivalent optimization problem as follows. Firstly, for a given set of weights, we have the following due to triangle inequality:
\begin{equation}
    |G_1| + |G_2| + \hdots + |G_l| \geq  |G_1 + \alpha_2 G_2 + \hdots + \alpha_l G_l |,
\end{equation}
where $\alpha_j\in\mathbb{C}$ with $|\alpha_j|=1$, and equality occurs iff the following co-phasing relation holds: $\arg{(G_1)} = \arg{(\alpha_2 G_2)} = \hdots = \arg{(\alpha_l G_l)}$. Hence, we have
\begin{align}
   \nonumber &\max_{w_i} |G_1| + |G_2| + \hdots + |G_l| \\
   \nonumber &= \max_{w_i} \max_{\substack{|\alpha_j|=1} } |G_1 + \alpha_2 G_2 + \hdots + \alpha_l G_l | \\
    &= \max_{\substack{|\alpha_j|=1} } \max_{w_i}  |G_1 + \alpha_2 G_2 + \hdots + \alpha_l G_l |.
\end{align}
The advantage of the equivalent problem is that for a given tuple $ (\alpha_2, \hdots, \alpha_l)$, the inner maximization problem takes the form $\max_{w_i} |\sum_i w_i y_i |$, because $G_1 + \alpha_2 G_2 + \hdots + \alpha_l G_l = \sum_i w_i(z_i^{(1)} + \alpha_2 z_i^{(2)} + \hdots + \alpha_l z_i^{(l)} )$, and thus it can be solved optimally using OPA. Therefore, our goal is to find the optimal tuple $(\alpha_2, \hdots, \alpha_l)$ for the outer maximization. 

A naive approach to solve this problem would be to choose a collection $\mathcal{D}$ of tuples $(\alpha_2, \hdots, \alpha_l)$, compute $\max_{w_i} |G_1 + \alpha_2 G_2 + \hdots + \alpha_l G_l |$ for each tuple in $(\alpha_2, \hdots, \alpha_l) \in \mathcal{D}$, and choose the $\{w_i\}$ that corresponds to the largest value. 

Instead, we present an iterative approach to solve this problem, based on the following theorem:
\begin{theorem}
   Let $\{\alpha^{(k)} \}_{k \geq 1} =  \{ (\alpha_2^{(k)}, \hdots, \alpha_l^{(k)}) \}_{k \geq 1}$ be a sequence of tuples, and $\{ w^{k} \}_{k \geq 1}$ be a sequence of weights that satisfy the following relations ($k \in \mathbb{N}$ denotes iteration number):
   \begin{itemize}
       \item $w^{(k)} = \arg \max_{w_i} |G_1 + \sum_{j=2}^l \alpha_j^{(k)} G_j|$.
       \item $\alpha_j^{(k+1)} = e^{j \big( \arg(G_1) - \arg(G_j) \big)}$ for $j = 2, \hdots, l$, where each $G_j$ is computed using the weights $w^{(k)}$.
   \end{itemize}
   Let $c_k = \max_{w_i} |G_1 + \sum_{j=2}^l \alpha_j^{(k)} G_j|$, and $d_k = |G_1(w^{(k)} )| + \hdots + |G_l(w^{(k)})|$. Then, the sequences $\{c_k\}_{k \geq 1}$ and $\{d_k\}_{k \geq 1}$ are both non-decreasing.
\end{theorem}

\begin{proof}
    The proof is presented in Appendix \ref{theorem_3_proof}.
\end{proof}

Thus, performing the iteration steps as specified in the above theorem for each $(\alpha_2, \hdots, \alpha_l) \in \mathcal{D}$ will yield better weights than the naive approach. The multiple beamforming method is summarized in Algorithm \ref{multibeam}.

\begin{algorithm}[htb]
    
    \caption{Multiple Beamforming} \label{multibeam}
    \begin{algorithmic}[1]
        \Require Directions $(\theta_1, \phi_1), \hdots, (\theta_l, \phi_l)$ in which we desire to beamform, and a discrete collection $\mathcal{D}$ of tuples $(\alpha_2, \hdots, \alpha_l)$, where $|\alpha_j| = 1$ for all $j\in \{2, \hdots, l\}$. 
        \vspace{2mm}
        \For{ each $(\alpha_2, \hdots, \alpha_l) \in \mathcal{D}$ }
            \State Set $(\alpha_2^{(1)}, \hdots, \alpha_l^{(1)}) = (\alpha_2, \hdots ,\alpha_l) $, and $ k \gets 1$.

            \While{convergence is not reached}

            \State Compute $w^{(k)} = \arg \max_{w_i} |G_1 + \sum_{j=2}^l \alpha_j^{(k)} G_j|$.

            \State Set $\alpha_j^{(k+1)} \gets e^{j \big( \arg(G_1) - \arg(G_j) \big)}$ for $j = 2, \hdots , l$, where each $G_j$ is computed using the weights $w^{(k)}$.
            \EndWhile

            \State Compute $|G_1| + \hdots + |G_l|$ using the last computed weights.

        \EndFor
        
        \State \Return weights that corresponds to largest $|G_1| + \hdots + |G_l|$.
    \end{algorithmic}
\end{algorithm}

We remark that while the above algorithm is not guaranteed to return the optimal weights, in practice the resulting weights are found to give very good performance. Typically, a good choice for $\mathcal{D}$ is the collection of all tuples $(\alpha_2,\hdots,\alpha_l)$ where each $\alpha_j$ is uniformly sampled on the unit circle. We discuss these issues and numerical results for multiple beamforming in Section III. We note in passing that all the algorithms that have been discussed so far can also be applied for beamforming from a basestation in case it also employs discrete weights; the only modification required is to eliminate the incident field related terms in \eqref{phi_mn}, which is conveniently achieved by setting $\theta_{in}=0$.

\section{Results for optimal beamforming}

In this Section, we provide various simulation results in support of the efficacy of our proposed method. In all simulations, we take the inter-element spacing as $d = \frac{\lambda}{2}$.

\subsection{Optimality of OPA}
We first compare the proposed OPA with the existing thresholding method. The hypothesis is that the thresholding method does not always return the optimal weights, and we provide an example that shows exactly that. Consider a $3 \times 3$ array, with incident direction $(\theta_{in}, \phi_{in}) = (-45^\circ, 215^\circ)$ and the desired direction $(\theta_0, \phi_0) = (-30^\circ, 35^\circ)$. All weights are constrained to the set $\{1,-1\}$. From the weights returned by the thresholding method, we obtain $|G(\theta_0, \phi_0)|^2 = -3.86 $ dB, whereas from the weights obtained by OPA, we obtain $|G(\theta_0, \phi_0)|^2 = -2.95$ dB. Fig.~\ref{comparision_with_thresholding} shows the scatter plots the weight assignment for each method. The OPA results are clearly superior, and this example confirms our hypothesis. 

\begin{figure}[htb!]
    \centering
    \includegraphics[width = 0.5\textwidth]{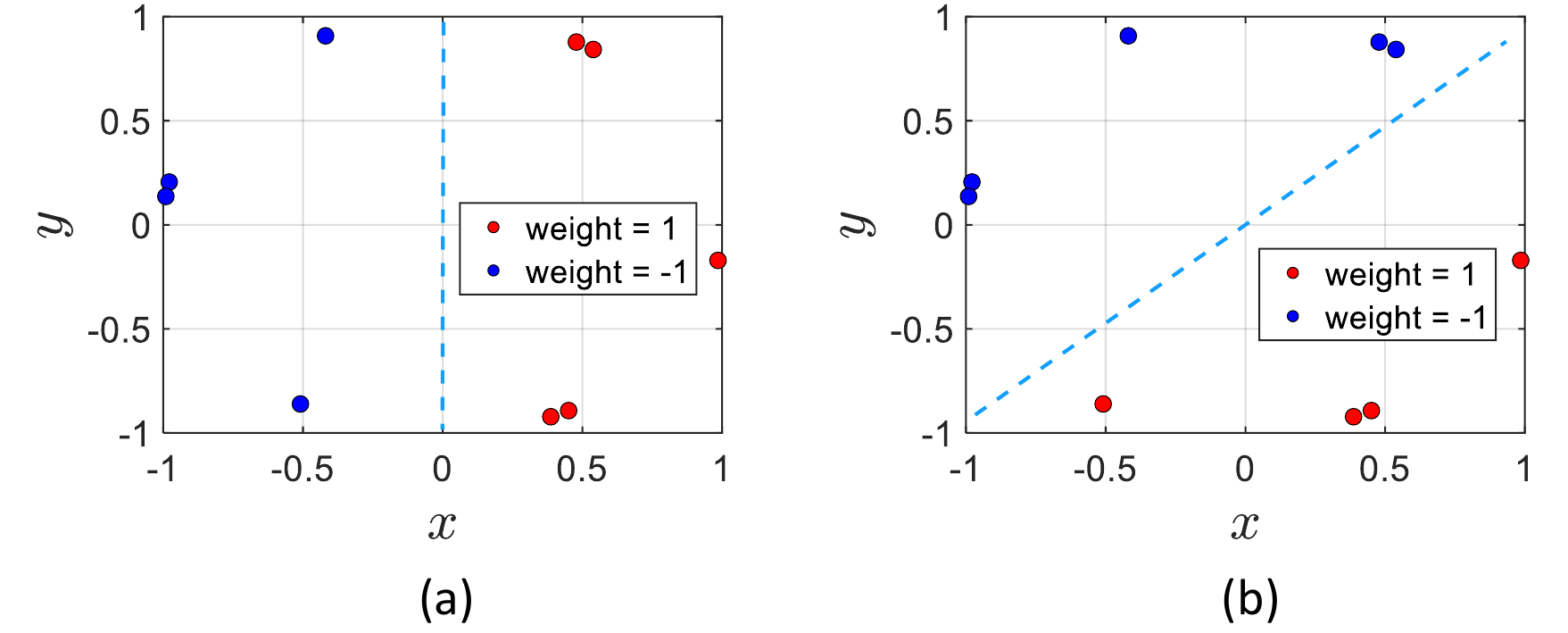}
    \caption{Scatter plots of $z_{m,n}$ along with assignment of weights due to (a) thresholding method (b) OPA. The dashed lines indicate the separating line in each case.}
    \label{comparision_with_thresholding}
\end{figure}

Next, we establish the optimality of OPA by comparing it with brute force combinatorial search of the discrete weights. For the same incident and desired directions as above, we compare thresholding, OPA and combinatorial search for various array sizes, as shown in Fig.~\ref{fig:thresh_vs_opa_vs_brute}. In Fig.~\ref{fig:thresh_vs_opa_vs_brute}(a), all weights are constrained to $\{ -1,1 \}$, while in Fig.~\ref{fig:thresh_vs_opa_vs_brute}(b), the amplitudes and phase difference for the constraint set of each weight are independently and randomly sampled from $U(0.75,1)$ and $U(160^\circ, 180^\circ)$ respectively, where $U$ denotes the uniform distribution (\ie,~for each weight $w_{m,n}$, we have $a_{m,n} \sim U(0.75, 1), |b_{m,n}| \sim U(0.75,1)$ and $\arg{b_{m,n}} \sim U(160^\circ, 180^\circ)$), in order to capture realistic scenarios (for e.g.~on account of fabrication related imperfections). Note that in case (b), the thresholding algorithm is modified slightly, such that the continuous solution $w_{m,n}$ is mapped to whichever among $\{ a_{m,n}, b_{m,n} \}$ it is closest to, for all $m,n$. We observe that OPA and combinatorial search yield the same values for $|G|^2$ in both cases, confirming the optimality of OPA. Further, the gap between the thresholding and OPA solutions is on average larger in (b) than in (a), which signifies a preference of OPA over thresholding in realistic scenarios.

\begin{figure}[htb!]
    \centering
    \includegraphics[width = 0.5 \textwidth]{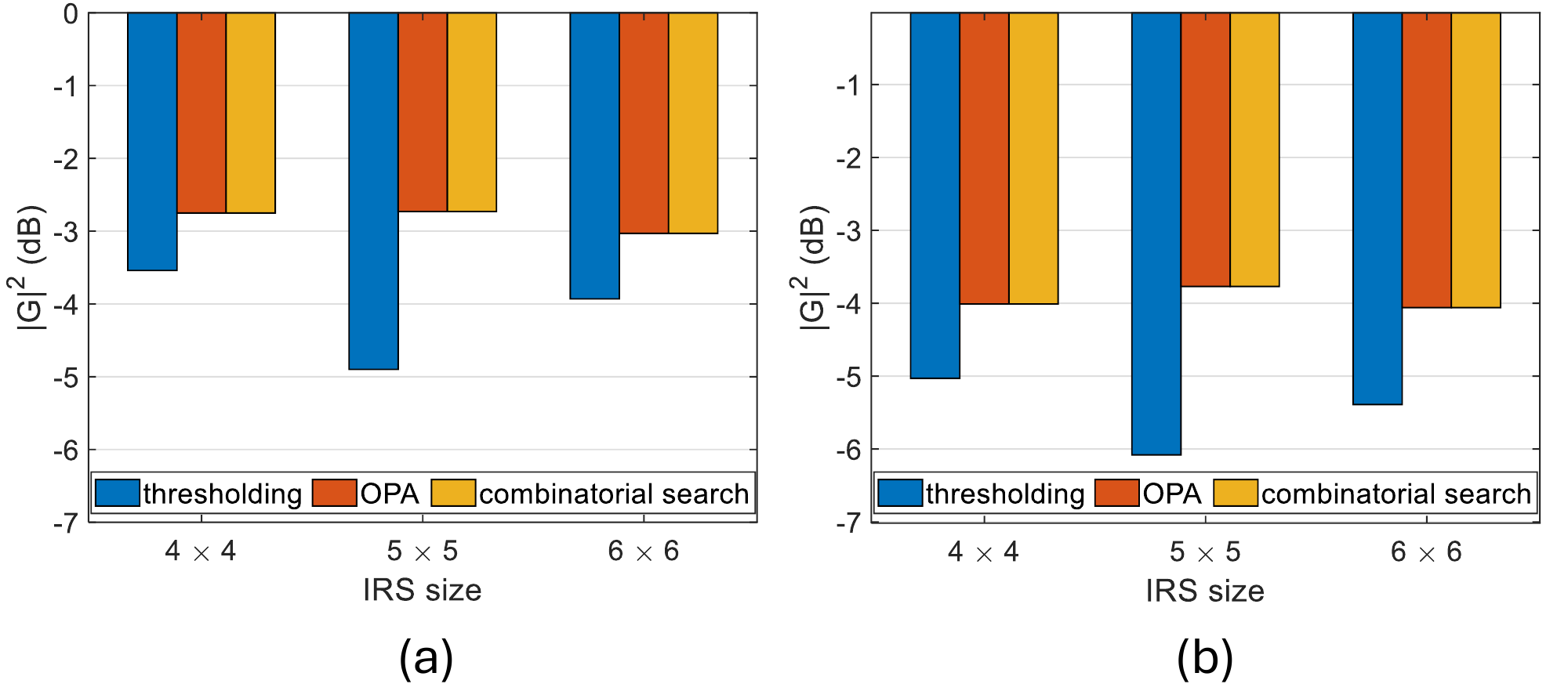}
    \caption{Comparision of thresholding, OPA and combinatorial search for various array sizes, with constraint sets $\{a_{m,n}, b_{m,n} \}$ equal to (a) $\{1, -1 \}$, and (b) $a_{m,n}, |b_{m,n}| \sim U(0.75,1) \text{ and } \arg{b_{m,n}} \sim U(160^\circ, 180^\circ) $, where $U$ denotes uniform distribution}
    \label{fig:thresh_vs_opa_vs_brute}
\end{figure}

\subsection{Visualization of weights}
We now visualize the output of the OPA by plotting the weight distribution across the IRS, and the resulting beam patterns for a couple of sample cases. We fix the array size to be $30 \times 30$, and choose $\theta_{in} = -30^\circ$, $\theta_0 = -15^\circ$ and $\phi_0 = 45^\circ, \phi_{in}=225^{\circ}$. OPA is run in MATLAB, on a system with Intel Core i7 processor and $16$ GB RAM, and the optimal weights are obtained in $0.4$ seconds. 

As seen in Fig.~\ref{fig:weight_dist_rad_pat}(a), the distribution consists of diagonal strips of $1$ and $-1$, whose angle with the $y$ axis (vertical axis) is precisely $\phi_{in}$. Intuitively, this makes sense because the strips are precisely the locus of points on the IRS that get the same initial phase due to the incident beam. 

Next, we consider the case of $\theta_{in} = -45^\circ$, $\theta_0 = -30^\circ$ and $\phi_0 = 0^\circ, \phi_{in}=180^{\circ}$. As seen in Fig.~\ref{fig:weight_dist_rad_pat}(a), the distribution consists of vertical strips of $1$ and $-1$ as expected. However, in addition to the expected reflected beam at $\theta_0 = -30^\circ$, we also observe a second peak of equal strength around $\theta_0 = -5^\circ$. Understanding the existence of this second peak is of great importance, and we do so in Section IV. 

\begin{figure}[htb!]
    \centering
    \includegraphics[width = 0.5 \textwidth]{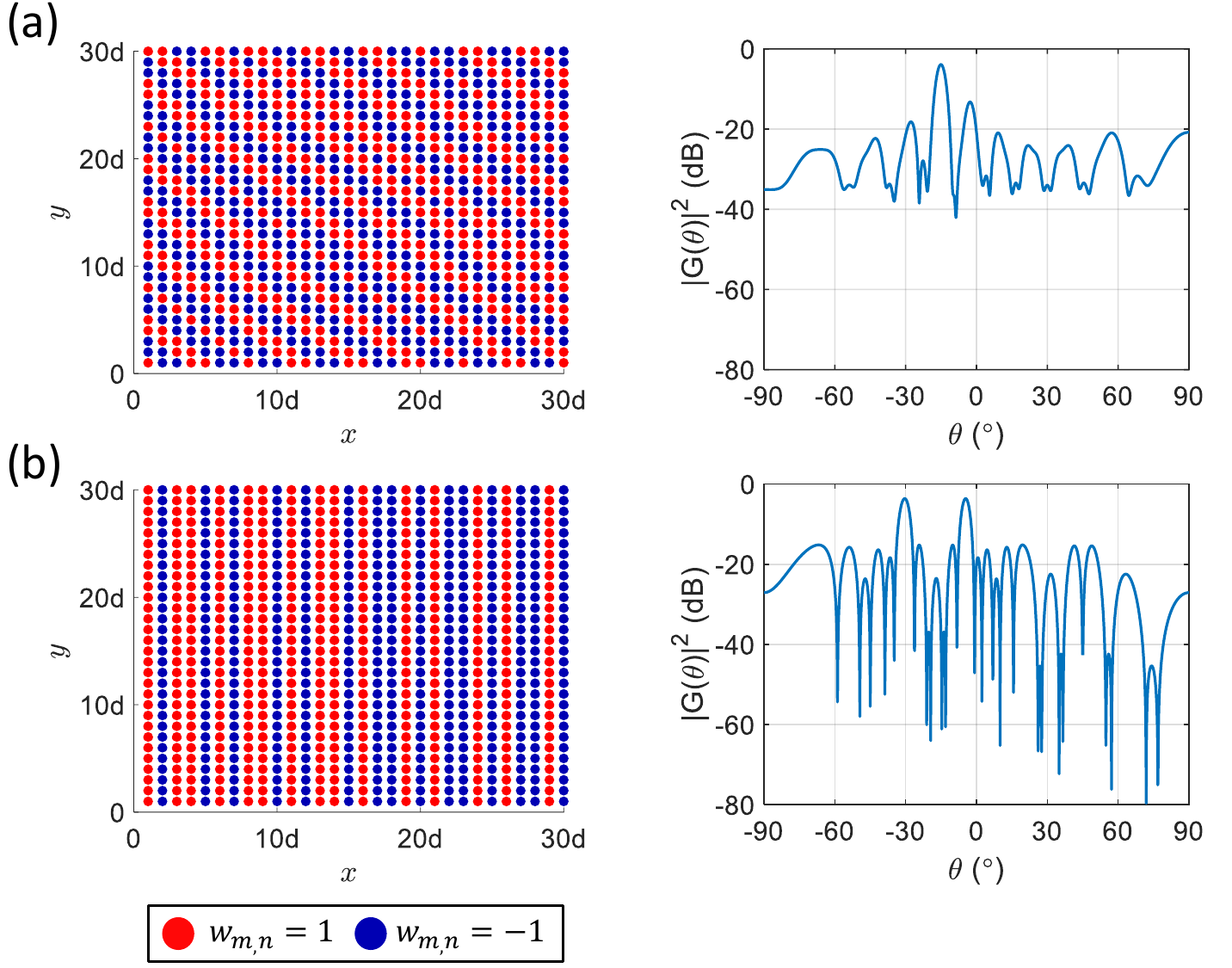}
    \caption{Distribution of optimal weights on IRS and the corresponding normalized radiation pattern, for (a) $(\theta_{in}, \phi_{in}) = (-30^\circ, 225^\circ)$ and $(\theta_0, \phi_0) = (-15^\circ, 45^\circ)$ (b)  $(\theta_{in}, \phi_{in}) = (-45^\circ, 180^\circ)$ and $(\theta_0, \phi_0) = (-30^\circ, 0^\circ)$}
    \label{fig:weight_dist_rad_pat}
\end{figure}

\subsection{Beamforming error vs number of bits}
An important consideration in discrete beamforming is the selection of number of bits of allowed phase shifts. While increasing the number of bits increases the quality of beamforming, such as directivity and power in the desired direction, it also makes it significantly harder to implement in hardware. This leads to a natural tradeoff between the number of bits and the quality of beamforming, and is explored in \cite{4476079_no_of_bits, zhang_reconfigurable_2020}. Of particular interest to us is the effect of number of bits on the \emph{beamforming error} (BE), which we define as the angle between the desired mainlobe direction, $(\theta_0, \phi_0)$, and the actual mainlobe direction, $(\theta_0', \phi_0')$, achieved by our algorithms, and is computed using the formula,
\begin{align}
    \label{beamforming_error_formula}
   \cos(\text{BE}) =  \sin{\theta_0} \sin{\theta_0'} \cos{ (\phi_0 - \phi_0')} + \cos{ \theta_0} \cos{\theta_0'}.
\end{align}

Fig.~\ref{fig:error_vs_bits} shows the variation of beamforming error with the number of bits, for different IRS sizes. We fix $\theta_{in} = 0$ (normal incidence), $ \phi_0 = 30^\circ$, and compute the average beamforming error over $\theta_0 \in [-45^\circ, 45^\circ]$ spaced by intervals of $ 5^\circ$. The optimal weights are obtained using kOPA. We observe that the beamforming error decreases if we increase the number of bits and the IRS size. In particular, for $15\times 15$ and $30 \times 30$ IRS, we observe negligible beamforming error even with a 1-bit implementation. Additionally, the 3-dB beamwidths and mainlobe gains for various array sizes are shown in Table.~\ref{table:beamwidth_and_gain}, where we define the gain in a direction $(\theta, \phi)$ as $20 \log{\big(MN |G(\theta, \phi)| \big)} $, with $G(\theta, \phi)$ given by \eqref{afac}. We observe that larger array sizes correspond to larger gains and smaller beamwidth. For example, a $30\times 30$ array exhibits sub-$5^\circ$ beamwidth. This, when combined with its low beamforming error, makes $30\times 30$ a good practical choice of array size. 

\begin{figure}[htb!]
    \centering
    \includegraphics[width = 0.4 \textwidth]{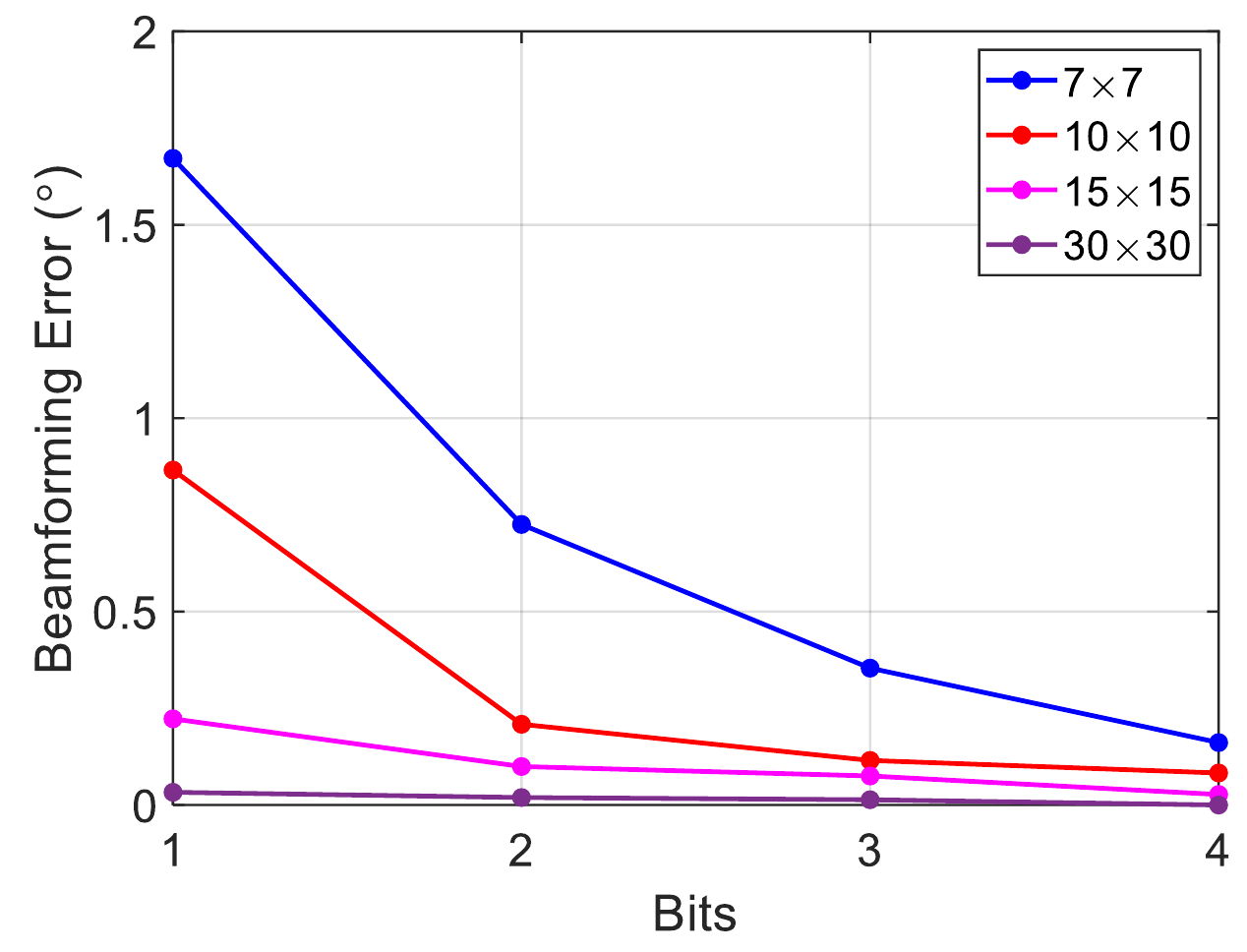}
    \caption{Variation of beamforming error with the number of bits of allowed phase shifts, under normal incidence with $\phi_{0} = 30^\circ$, for $7\times 7, 10\times 10$, $15\times 15$ and $30 \times 30 $ IRS. The errors are averaged over $\theta_0 \in [-45^\circ, 45^\circ]$ spaced by intervals of $5^\circ$. }
    \label{fig:error_vs_bits}
\end{figure}

\begin{table}[htb!]

\caption{3-dB beamwidth and mainlobe gain for various array sizes}
\label{table:beamwidth_and_gain}
\centering
\begin{tabular}{|c|c|c|c|c|}
    \hline
    & $7 \times 7$ & $10 \times 10$ & $15 \times 15$ & $30 \times 30$ \\ \hline
    3-dB beamwidth & $15.8^\circ$  & $10.3^\circ$    & $6.8^\circ $    & $3.4^\circ$    \\ \hline
    $20 \log(MN |G|)$ & 31    & 36.4    & 43.7    & 55.3    \\ \hline
\end{tabular}
\end{table}

\subsection{Beamforming in multiple directions}
We now illustrate an example of multiple beamforming for various array sizes. Fig.~\ref{fig:multibeam}(a) and \ref{fig:multibeam}(b) show the radiation patterns for various array sizes, for the cases where we desire two beams and three beams, respectively. The weights are constrained to $\mathcal{S} = \{1, -1\} $, and are obtained using Algorithm \ref{multibeam}. For the two-beam case in Fig.~\ref{fig:multibeam}(a), we choose $\mathcal{D} = \{ e^{ j\frac{2 \pi k}{30} } \big|  k = 1, \hdots, 30\}$, and for the three-beam case in Fig.~\ref{fig:multibeam}(b), we choose $\mathcal{D} = \{ \big( e^{ j\frac{2 \pi k}{30}},  e^{ j\frac{2 \pi l}{30} }\big) \big|  k,l = 1, \hdots, 30\}$. For each tuple $\alpha \in \mathcal{D}$, the iteration process in Algorithm \ref{multibeam} is observed to converge very quickly (within $10$ iterations).

We observe that for smaller array sizes such as $10\times 10$, the ability to form multiple beams is limited, as evident from the poor sidelobe levels. Additionally, smaller array sizes perform poorly when the desired angles are close to each other. This issue is not faced by larger array sizes such as $30 \times 30$, where we can achieve a sidelobe level of roughly $10$ dB as seen in Fig.~\ref{fig:multibeam}, providing yet another justification to use larger array sizes.

\begin{figure}[htb!]
    \centering
    \includegraphics[width=0.5\textwidth]{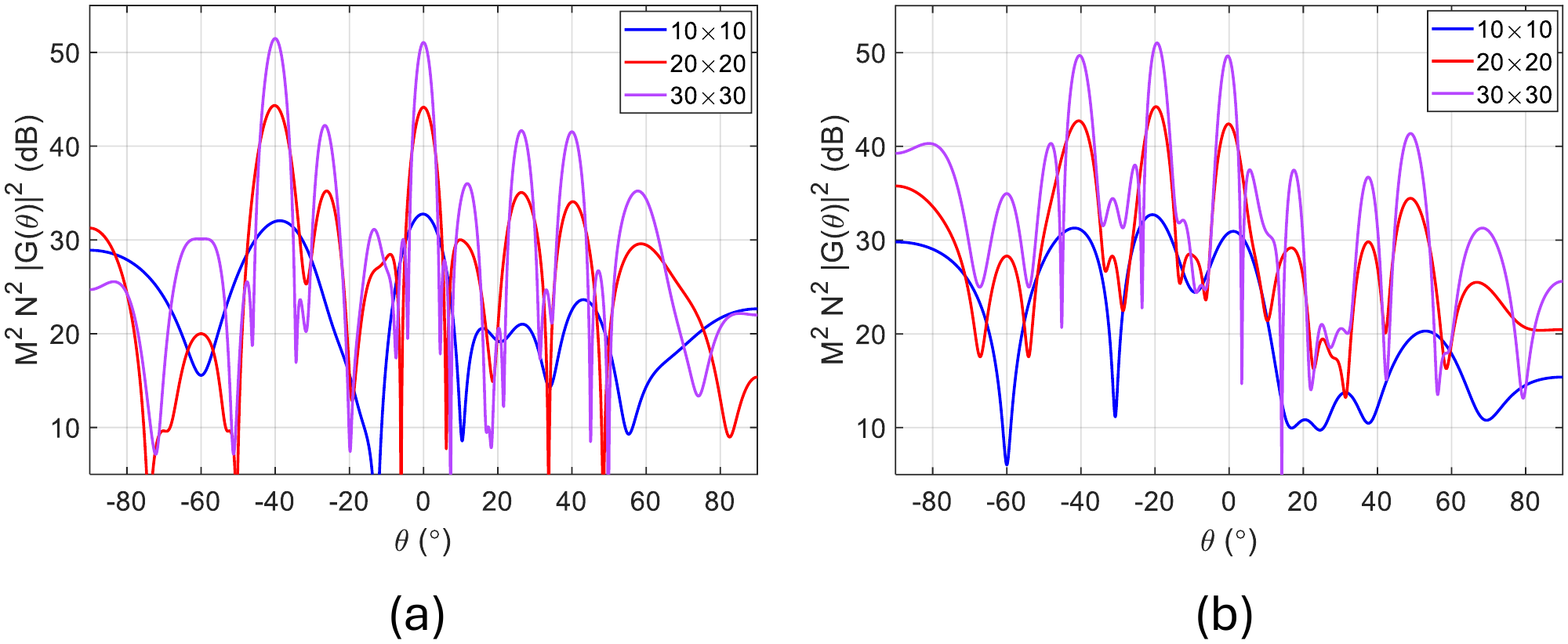}
    \caption{Radiation patterns illustrating multiple beamforming for various array sizes, with incident direction $(\theta_{in}, \phi_{in}) = (60^\circ, 210^\circ)$, for (a) two beams, where $(\theta_1, \phi_1) = (0^\circ, 30^\circ )$ and $(\theta_2, \phi_2) = (-40^\circ, 30^\circ )$ (b) three beams, where $(\theta_1, \phi_1) = (0^\circ, 30^\circ ), (\theta_2, \phi_2) = (-20^\circ, 30^\circ )$ and $(\theta_3, \phi_3) = (-40^\circ, 30^\circ )$.  }
    \label{fig:multibeam}
\end{figure}

\section{Grating lobe Mitigation}
\subsection{Analysis of grating lobes}
We now investigate a critical issue in beamforming concerning the presence of grating lobes in the radiation pattern. A grating lobe is an undesired lobe with the same intensity as the main lobe (\ie,~has sidelobe level (SLL) equal to 0 dB). Such a lobe is seen in Fig.~\ref{fig:weight_dist_rad_pat}(b); these lobes can lead to interference in the unwanted directions, and security issues in the context of wireless communications and are undesirable.

We now derive theoretical conditions under which a grating lobe is guaranteed to exist, for the 1-bit case where each weight is constrained to $\{1, -1\}$. By examining the array factor expression in \eqref{afac}, we can see that given a main lobe at $(\theta_0, \phi_0)$, a grating lobe will appear at $(\theta^*, \phi^*)$ if the following condition holds for all $m,n$:
\begin{equation}
    \label{gl_cond}
    \varphi_{m,n}(\theta^*,\phi^*) = -\varphi_{m,n}(\theta_0, \phi_0) + \pi a m + \pi b n,
\end{equation}
where $a$ and $b$ are even integers, because \eqref{gl_cond} would then imply $|G(\theta^*, \phi^*)| = |G(\theta_0, \phi_0)|$. Further, using the phase expression from \eqref{phi_mn} and comparing the coefficients of $m,n$ on both sides, \eqref{gl_cond} gives:
\begin{subequations} \label{conditions_for_0_sll}
    \begin{align}
    &\underbrace{- 2 \sin{\theta_{in}} \cos{\phi_{in}} + \sin{\theta_0} \cos{\phi_0} + \frac{\lambda a}{2d}}_{A} = - \sin{\theta^*} \cos{\phi^*}\\  
    &\underbrace{- 2 \sin{\theta_{in}} \sin{\phi_{in}} + \sin{\theta_0} \sin{\phi_0} + \frac{\lambda b}{2d}}_{B} = - \sin{\theta^*} \sin{\phi^*} 
\end{align}
\end{subequations}

If \eqref{conditions_for_0_sll} holds for some choice of even integers $a$ and $b$, then we have

\begin{equation}
    A^2 + B^2 = \sin^2{(\theta^*)} \text{,  and } \tan{\phi^*} = \frac{B}{A}.
\end{equation}

Therefore, if there exist even integers $a,b$ such that $A^2 + B^2 \leq 1$, then we can compute $\theta^*=\pm \arcsin{\sqrt{A^2+B^2}}$ and $\phi^* =\arctan{\frac{B}{A}}$. If it turns out that $(\theta^*, \phi^*) \neq (\theta_0, \phi_0)$, then a grating lobe exists at $(\theta^*, \phi^*)$.  

\begin{figure}[htb!]
    \centering
    \includegraphics[width = 0.5 \textwidth]{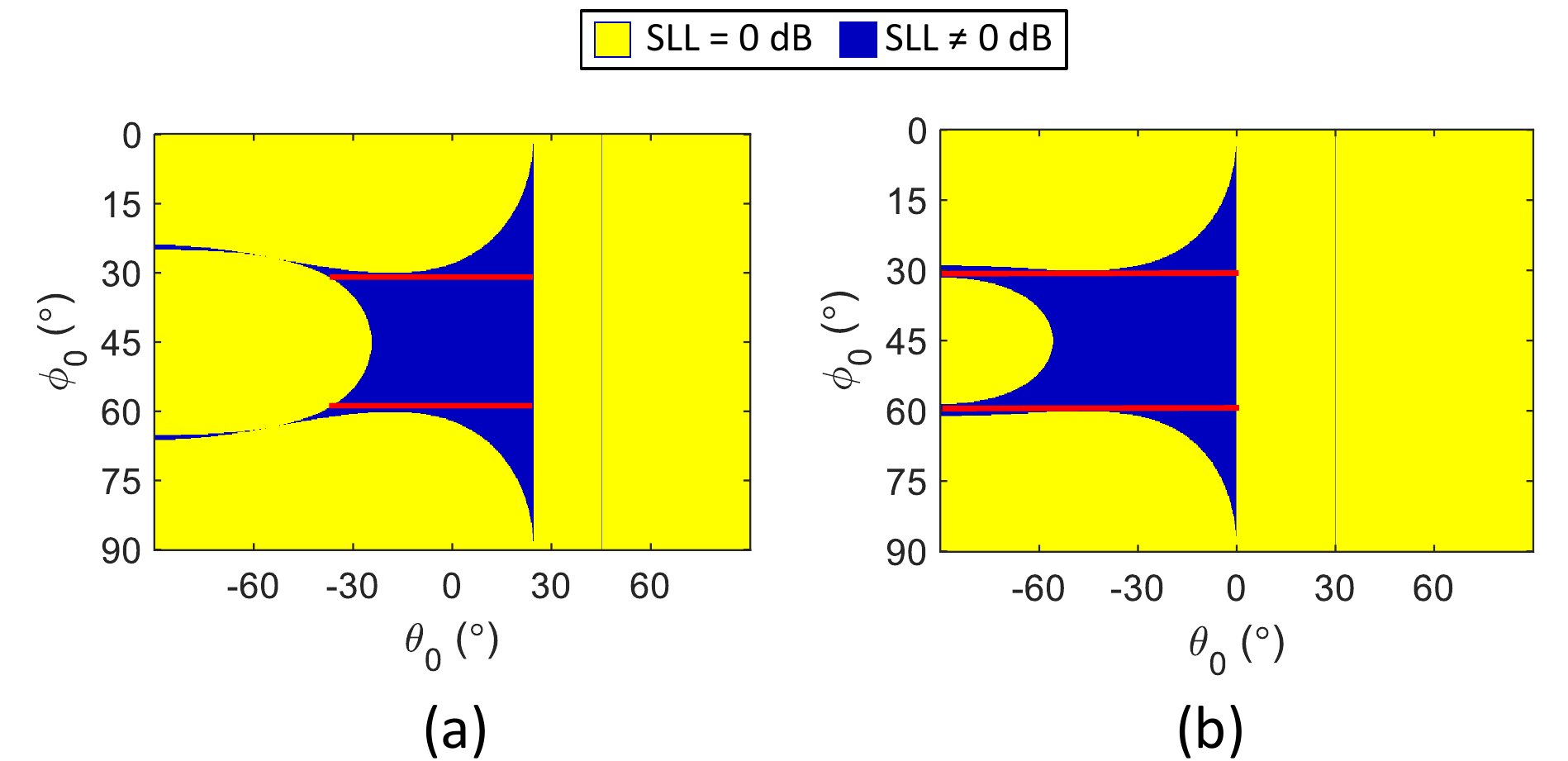}
    \caption{Plots showing the existence of grating lobes according to \eqref{conditions_for_0_sll} for varying main lobe directions $(\theta_0, \phi_0)$, when the incident and reflected beams lie on the same vertical plane ($\phi_{0} = \phi_{in} - 180^\circ$), for (a) $\theta_{in} = -45^\circ$, and (b) $\theta_{in} = -30^\circ$. Yellow regions indicate the presence of a grating lobe at $(\theta^*, \phi^*) $ computed from \eqref{conditions_for_0_sll}, different from the given main lobe direction $(\theta_0, \phi_0)$. The red lines indicate the values of $\phi_0$ corresponding to largest range of non-zero SLL.}
    \label{theoretical_heatmaps_for_45_and_30}
\end{figure}

\begin{figure}[htb!]
    \centering
    \includegraphics[width = 0.5\textwidth]{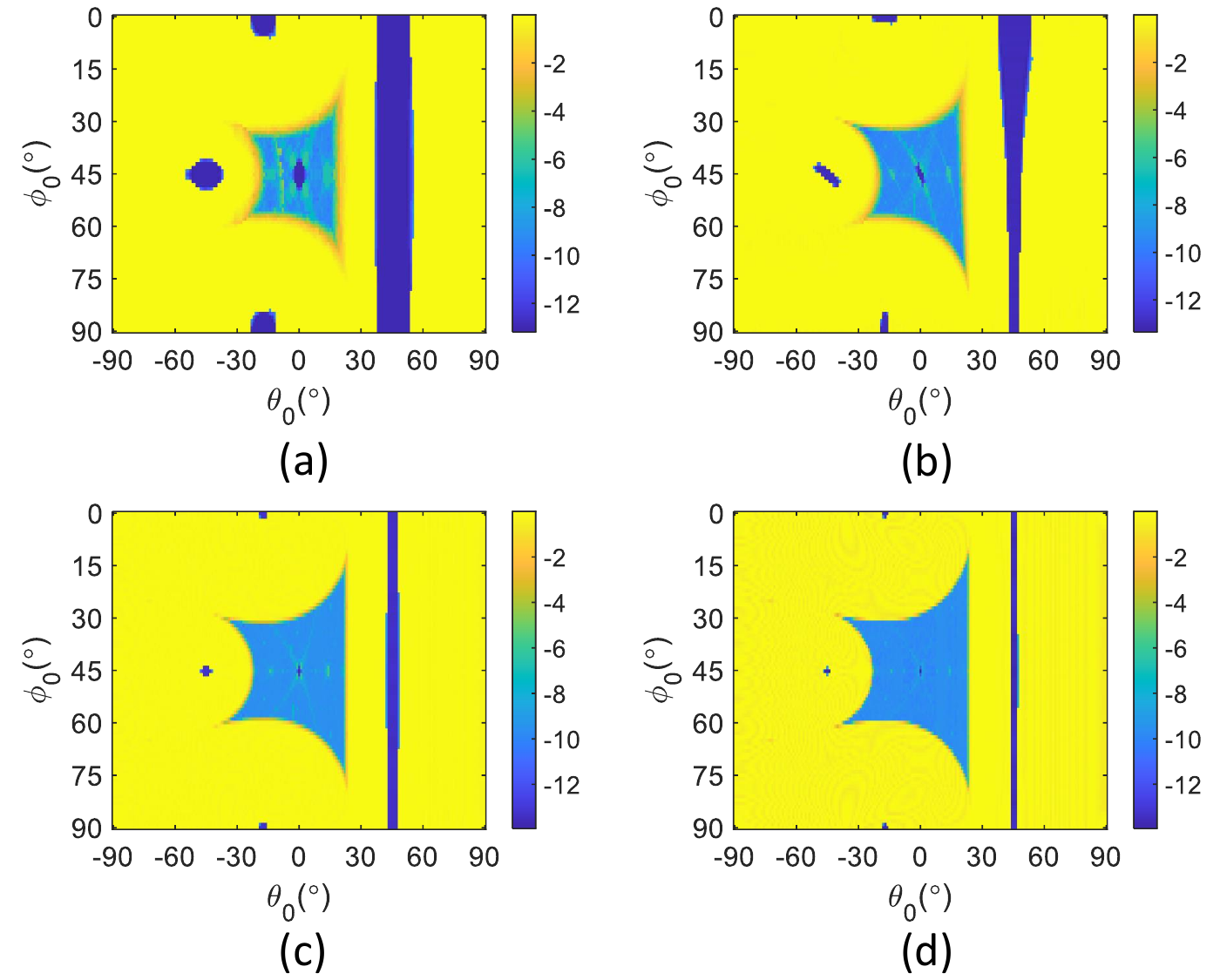}
    \caption{Plots indicating the simulated SLL (in dB) for varying desired reflected directions $(\theta_0, \phi_0)$, when the incident and desired reflected directions lie on the same vertical plane ($\phi_{0} = \phi_{in} - 180^\circ$) and $\theta_{in} = -45^\circ$, for array sizes: (a) $10 \times 10$, (b) $10 \times 30$, (c) $30 \times 30$, and (d) $50 \times 50$.}
    \label{simulated_45}
\end{figure}

In order to graphically visualize the algebraic conditions for the existence of a grating lobe (\ie,~\eqref{conditions_for_0_sll}), we consider several cases. For ease of visualization, we fix the incidence beam angle and sweep the required output beam angle, while maintaining the incident and required output wavevectors to be in the same vertical plane. The plots in Fig.~\ref{theoretical_heatmaps_for_45_and_30} show the region (colored yellow) of $(\theta_0, \phi_0)$ for which a grating lobe exists at some $(\theta^*, \phi^*) \neq (\theta_0, \phi_0)$, according to \eqref{conditions_for_0_sll}. The blue regions indicate the absence of grating lobes. Note that $(-\theta, \phi)$ is equivalent to $(\theta, \phi + \pi)$; thus the left half of the plots, where $\theta_0 <0$, denote the reflected beam going in the backward scattering direction. We emphasize that the conditions in \eqref{conditions_for_0_sll} are irrespective of the size of the IRS; it gives the condition for the existence of a grating lobe if the main lobe is formed \emph{exactly} at $(\theta_0, \phi_0)$.

We verify the theoretical plots in Fig.~\ref{theoretical_heatmaps_for_45_and_30} by performing simulations for various arrays, where the optimal weights are obtained using OPA. Fig.~\ref{simulated_45} shows the simulated SLL plots for various array sizes, and they increasingly agree with the theoretical plot in Fig.~\ref{theoretical_heatmaps_for_45_and_30}(a) as the array sizes grow. We observe bands around $\theta_0 = 45^\circ$ in Fig.~\ref{simulated_45}, in addition to certain other small regions of non-zero SLL. This occurs due to beamforming error (defined in section III), which decreases with increasing array size. This can be seen by the bands becoming narrower.

Having identified the existence of grating lobes and confirmed them via simulations, we now seek to propose strategies to mitigate these lobes, as their presence is detrimental to the intended operation of an IRS. As we will demonstrate, the proposed OPA and gOPA algorithms are well suited to implement these mitigation strategies. For the sake of definiteness, we only consider 1-bit IRS implementations in this discussion. We remark that for $k > 2$ phase shifts, \eqref{gl_cond} does not necessarily imply a grating lobe at $(\theta^*, \phi^*)$, and in general a grating lobe does not occur.

\begin{figure}
    \centering
    \includegraphics[width = 0.45 \textwidth]{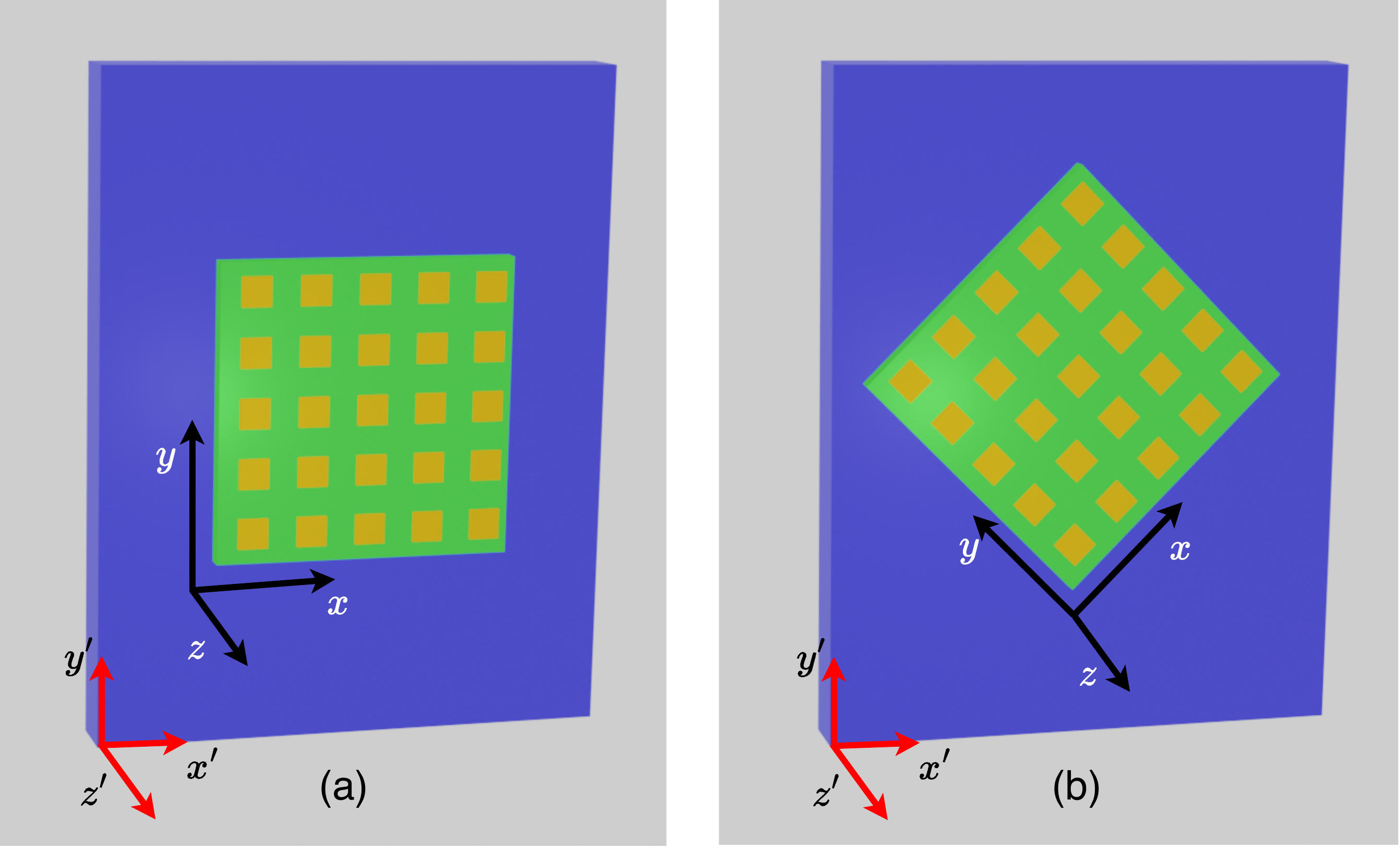}
    \caption{Two different alignment scenarios of a vertically deployed $5\times 5$ IRS  ($x-y$ plane) with respect to the ground ($x'-z'$ plane): (a) $\hat{x},\hat{x}'$ are parallel, and (b) $\hat{x},\hat{x}'$ are at an angle.}
    \label{fig:perspec}
\end{figure}

We consider a vertically oriented IRS (\ie,~with its surface normal parallel to the ground), and an incident beam emitted by a transmitter with the wavevector parallel to the ground. It is desired to beamform to a mobile user in the same horizontal plane, as depicted in Fig.~\ref{fig:perspec}. There arise two scenarios:
\subsection*{Array $x$-axis parallel with respect to the ground}
This amounts to $\phi_0=0^\circ$ and $\phi_{in} = 180^\circ$. As we can see from Fig.~\ref{theoretical_heatmaps_for_45_and_30} for $\phi_0=0$, a grating lobe always exists, regardless of the value of $\theta_{in}$. It is easy to see this theoretically, as shown below. \eqref{conditions_for_0_sll} simplifies to:
\begin{subequations}
\begin{align}
    2 \sin{\theta_{in}} + \sin{\theta_0} + a &= -\sin{\theta^*} \cos{\phi^*} \\
    b &= -\sin{\theta^*} \sin{\phi^*}.
\end{align}    
\end{subequations}
Clearly if $b = 0$ and $a$ is chosen to be $2,0$ or $-2$, depending on whether $(2 \sin{\theta_{in}} + \sin{\theta_0})$ lies in $[-3,-1]$, $[-1,1]$ or $[1,3]$ respectively, then a grating lobe exists for $\phi^* = 0$ and $\theta^* = \arcsin{ ( 2 \sin{\theta_{in}} + \sin{\theta_0} + a ) }$. The lobe doesn't exist only when $\theta^* = \theta_0$, which happens for two particular angles (one of them being $\theta_0 = -\theta_{in}$). Not only does there exist a grating lobe for almost all values of $\theta_0$, the grating lobe also lies along the $xz$ plane, as $\phi^* = 0$.

\subsection*{Array $x$-axis titled with respect to the ground}
When $\phi_{in}\neq 180^\circ$, the issue of grating lobes can be circumvented for some user locations. For e.g.~the red lines in Fig.~\ref{theoretical_heatmaps_for_45_and_30} indicate the maximum possible range of grating-lobe free operation for a few different values of $\theta_{in}$. In particular, we see that if $\theta_{in}=-30^\circ$, a user can move in the entire half-range $\theta_0=[0,90^\circ]$ for a titled IRS with $\phi_0=30^\circ$ and $\phi_{in} = 210^\circ$. Of course, this range is curtailed depending on the size of the array as seen in Fig.~\ref{simulated_45}.

It is often not convenient to tilt an IRS from an implementation perspective, because most deployments are likely to be on walls or sides of buildings, where a horizontal form-factor is more suitable. Therefore, we outline two strategies that mitigating the issue of grating lobes while aligning the IRS axis with the ground.

\subsection{Triangular lattice for IRS layout}
In all the examples considered so far, the layout of the constitutive elements was assumed to be rectangular (e.g.~see Fig.~\ref{fig:IRS_diagram}). Instead, if we construct an IRS by arraying its elements in a triangular lattice, it can be shown that $\phi_0=0^\circ$ and $\phi_{in} = 180^\circ$ becomes viable for beamforming without encountering the issue of grating lobes. 

The array factor is now modified as:
\begin{equation}
 G(\theta, \phi) = \frac{1}{MN}\sum_{m,n} w_{m,n} e^{j\frac{2\pi}{\lambda\,}r_{m,n}^T \big(-\hat{u}(\theta, \phi) + \hat{u}(\pi - \theta_{in}, \phi_{in})\big)},
\end{equation}
where $ \hat{u}(\theta, \phi) = [\sin{\theta} \cos{\phi}\,\, \sin{\theta}\sin{\phi} \,\,\cos{\theta}]^T$ and $r_{m,n} = m d_1 + n d_2$, with $0 \leq n \leq N$ and $ -\lfloor \frac{n}{2} \rfloor \leq m \leq M - \lceil \frac{n}{2} \rceil $. Here, $d_1=d[1\,\,0\,\,0]^T$ and $d_2=d[\frac{1}{2}\,\,\frac{\sqrt{3}}{2}\,\,0]^T$ are the basis vectors for the lattice.

Analogous to the analysis in the case of the rectangular lattice, we can work out the conditions under which a grating lobe appears. The relations for a grating lobe at $(\theta^*,\phi^*)$ are (details in Appendix \ref{triangular_app}):
\begin{subequations}
\label{grating_lobe_condition_for_triangular_lattice}
\begin{align} \nonumber
    -2 \sin{\theta_{in}} \cos{\phi_{in}} + \sin{\theta_0} \cos{\phi_0}& + \frac{\lambda a}{2d} \\
    &\, = -\sin{\theta^*} \cos{\phi^*}, \\  \nonumber
    -2 \sin{\theta_{in}} \sin{\phi_{in}} + \sin{\theta_0} \sin{\phi_0} + &\frac{\lambda}{2d}\Big(  \frac{-a + 2b}{\sqrt{3}}\Big) \\ 
    &\,= -\sin{\theta^*} \sin{\phi^*},
 \end{align}
\end{subequations}
where $a,b$ are even integers, as in \eqref{conditions_for_0_sll}.

\begin{figure}[htb!]
    \centering
    \includegraphics[width = 0.5\textwidth]{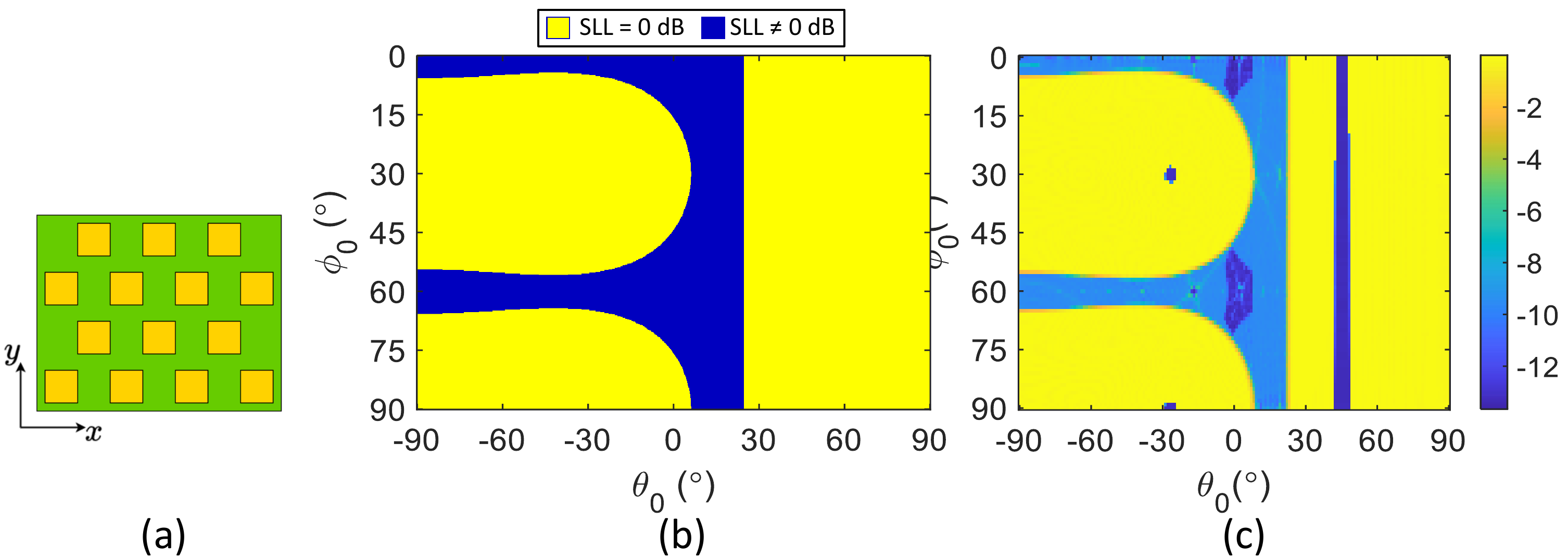}
    \caption{(a) A schematic of a 3×3 triangular lattice (b) Plot showing the existence of grating lobes according to \eqref{grating_lobe_condition_for_triangular_lattice} for varying main lobe directions $(\theta_0, \phi_0)$, when incident and reflected directions lie on the same vertical plane ($\phi_{in} = \phi_0 + 180^\circ $) and $\theta_{in} = -45^\circ$ (c) Corresponding plot indicating the simulated SLL (measured in dB) for varying desired reflected directions $(\theta_0, \phi_0)$, for a $30\times 30$ triangular IRS.}
    \label{triangular_theo_sim}
\end{figure}

Substituting $\phi_{in} = \phi_0 + \pi$ in these relations gives us the theoretical plot as shown in Fig.~\ref{triangular_theo_sim}(a), while Fig.~\ref{triangular_theo_sim}(b) shows the sidelobe levels after the application of OPA for a finite size IRS. A major difference between the SLL plots of the two different lattices is that $\phi_0  = 0^\circ$ gives a wide range of grating-lobe free operation for triangular lattices, which is not the case in rectangular arrays. A detailed theoretical analysis in support is provided in Appendix \ref{triangular_app}. Thus, the triangular lattice enables the alignment of the IRS axis with the ground while mitigating the grating lobe issue.

\subsection{Prephasing technique for elimination of grating lobes}

In the grating lobe considerations considered so far, a non-zero incidence angle was assumed. However, there may be deployment scenarios where keeping $\theta_{in}=0$ is desired, since this opens up the possibility of steering the reflected beam in either quadrant (\ie,~$\theta_0>0$ or $\theta_0<0$). 

However, this presents a fundamental problem, because when: (a) $\theta_{in} = 0$, and (b) the weights are constrained to $\{1, -1\}$, we can show that $|G(\theta, \phi)| = |G(-\theta, \phi)|$ for all $(\theta, \phi)$, irrespective of the type of lattice:
\begin{align}\label{conjugate_G}
    \nonumber G(-\theta, \phi) &= \frac{1}{MN}\sum_{m = 1}^M \sum_{n = 1}^N w_{m,n} e^{j \varphi_{m,n}(- \theta, \phi)} \\
    &= \frac{1}{MN}\sum_{m = 1}^M \sum_{n = 1}^N w_{m,n}^* \cdot  \big( e^{ j \varphi_{m,n}(\theta, \phi) } \big)^* = G(\theta, \phi)^*,
\end{align}
where $^* $ denotes complex conjugation. The key argument we used is that $w_{m,n} = w_{m,n}^*$, since each $w_{m,n} \in \mathbb{R}$. This symmetry means that there will always be a grating lobe for normal incidence. As earlier work \cite{yin_single-beam_2020} has identified, if the primary assumption of $w_{m,n} \in \mathbb{R}$ is attacked (for e.g.~by constraining some of the weights to $\{j, -j\}$), we arrive at a technique of symmetry breaking. This technique is called \emph{prephasing}. 

There are several techniques to implement prephasing. For instance, adding random phase perturbations to each element of the IRS also helps break the symmetry \cite{kashyap_mitigating_2020}.  Another approach \cite{10195171_quantization_suppression} involves shifting subsets of the IRS elements in the normal direction by a small amount. Yet another possibility is to choose a fraction $\kappa$ of tuples from $\{1, \dots, M\} \times \{1, \dots N\}$ (denoting this set as $\mathcal{Q}$), and to then solve the following optimization problem (indeed, \cite{yin_single-beam_2020} does so with $\kappa=0.5$):
\begin{align}
    \label{prephasing_problem}
     \max_{w_{m,n}} \quad &|G(\theta_0, \phi_0)| \\
    \nonumber \textrm{s.t.} \quad &w_{m,n} \in \{j, -j\},\hspace{4mm} \forall (m,n) \in \mathcal{Q}, \\
    \nonumber &w_{m,n} \in \{1,-1\}, \hspace{4mm} \forall (m,n) \notin \mathcal{Q}.
\end{align}
In general, there can be more than two prephases \cite{10179250_prephase_low_res_hybrid_phasing}. Formally, if the prephase set is $\{\psi_1, \dots \psi_p\}$, and if $\{ \mathcal{Q}_i \}_{i=1}^p$ is a partition of the IRS such that each unit cell in $\mathcal{Q}_i$ gets the prephase $\psi_i$, then the optimization problem becomes
\begin{align}
    \label{prephasing_problem_general}
    &\max_{w_{m,n}} \quad |G(\theta_0, \phi_0)| \\ \nonumber 
    &\textrm{s.t. } w_{m,n} \in \{e^{j \psi_i}, e^{j (\psi_i + \pi) }\},\, \forall (m,n) \in \mathcal{Q}_i, \,\forall i = 1,\dots,p.
\end{align}

Remarkably, this problem can be solved by the proposed gOPA, as is evident by a comparison of the above equation with the gOPA formulation in Algorithm 3. Even with the earlier techniques of random phase perturbations \cite{kashyap_mitigating_2020}, the problem of determining the optimum weights for beamforming can be solved by the gOPA.

\begin{figure}[htb!]
    \centering
    \includegraphics[width = 0.5 \textwidth]{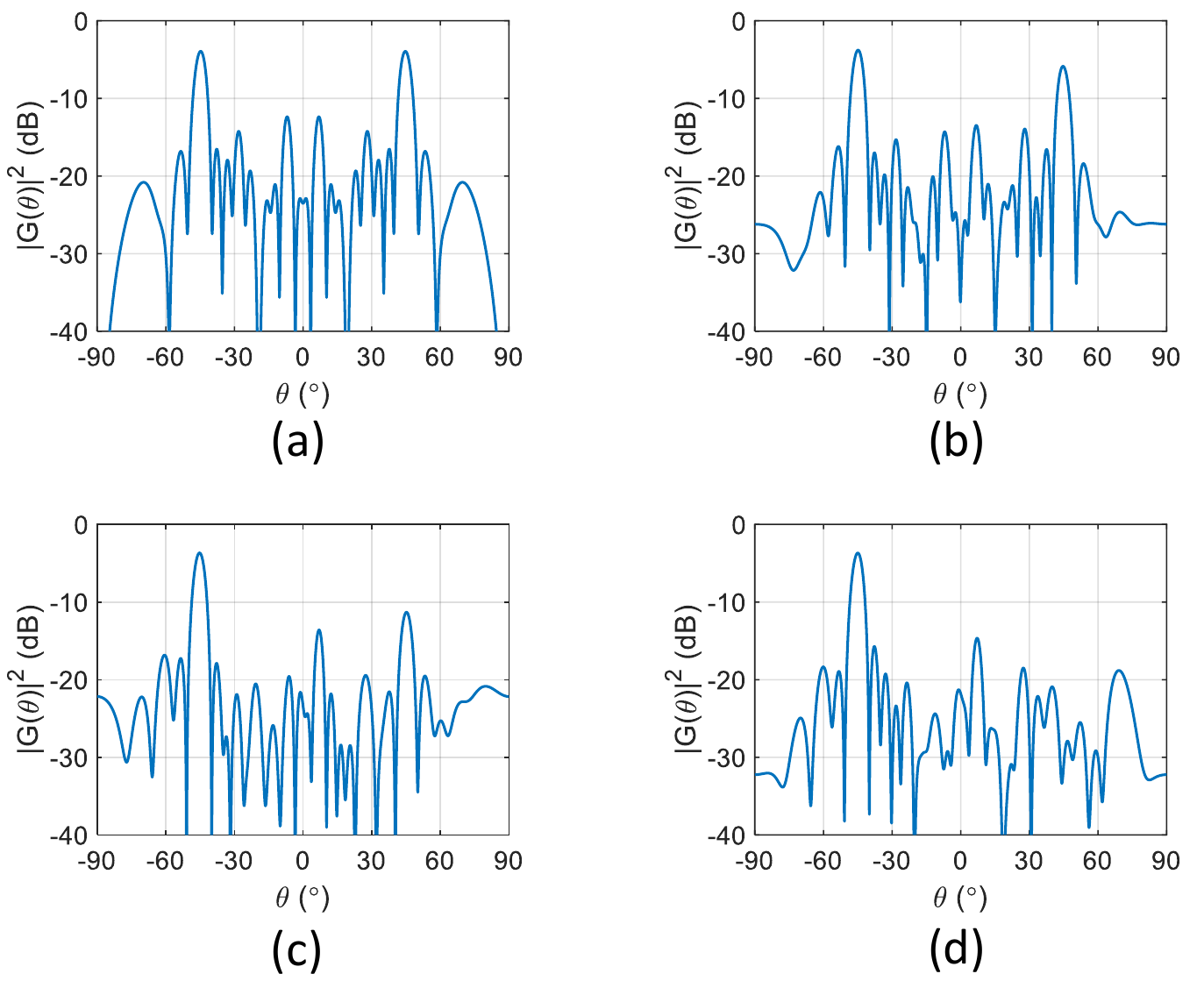}
    \caption{Radiation patterns for $\theta_{in} = 0^\circ$, $\theta_0 = -45^\circ$, $\phi_0 = 0^\circ, \phi_{in} = 180^\circ$ in a 30$\times$30 IRS: (a) $\kappa = 0$ (SLL = $0$ dB) (b) $\kappa = 0.1$ (SLL = $-2$ dB) (c) $\kappa = 0.3$ (SLL = $-7.6$ dB) (d) $\kappa = 0.5$ (SLL = $-10.9 $ dB)}
    \label{fig:radiation_patterns_prephasing}
\end{figure}

We implement the prephasing method as captured in \eqref{prephasing_problem}, and implement the gOPA for various values of the prephasing fraction $\kappa$. For each value of $\kappa $, the set $\mathcal{Q}$ is chosen randomly such that $| \mathcal{Q} | = \kappa MN$. 

Fig.~\ref{fig:radiation_patterns_prephasing} shows the radiation pattern in the $xz$ plane for several values of $\kappa$, with the highest SLL suppression of -10.9 dB achieved for $k=0.5$ in a $30\times 30$ IRS. We also note that while SLL reduces with increasing $\kappa$, the magnitude of the mainlobe remains the same. Thus, prephasing is a viable means of grating lobe mitigation while not sacrificing the mainlobe magnitude.

Next, we confirm the feasibility of scanning a single beam across space for (fixed) normal incidence. Fig.~\ref{fig:prephasing_scanning} shows the radiation patterns along the $xz$-plane, of a $30\times 30$ IRS in the scanning range $[-30^\circ, 30^\circ]$ spaced by intervals of $10^\circ$. The set $\mathcal{Q}$ is chosen randomly, with $\kappa = 0.5$. The worst-case SLL is observed to be $-8.6$ dB. We note that the SLL can be further improved by optimizing the set $\mathcal{Q}$ itself by using evolutionary algorithms \cite{9807634_prephase_1-bit_low-sll}, however at a significant computational cost compared to our proposed method. 

\begin{figure}[htb!]
    \centering
    \includegraphics[width = 0.5 \textwidth]{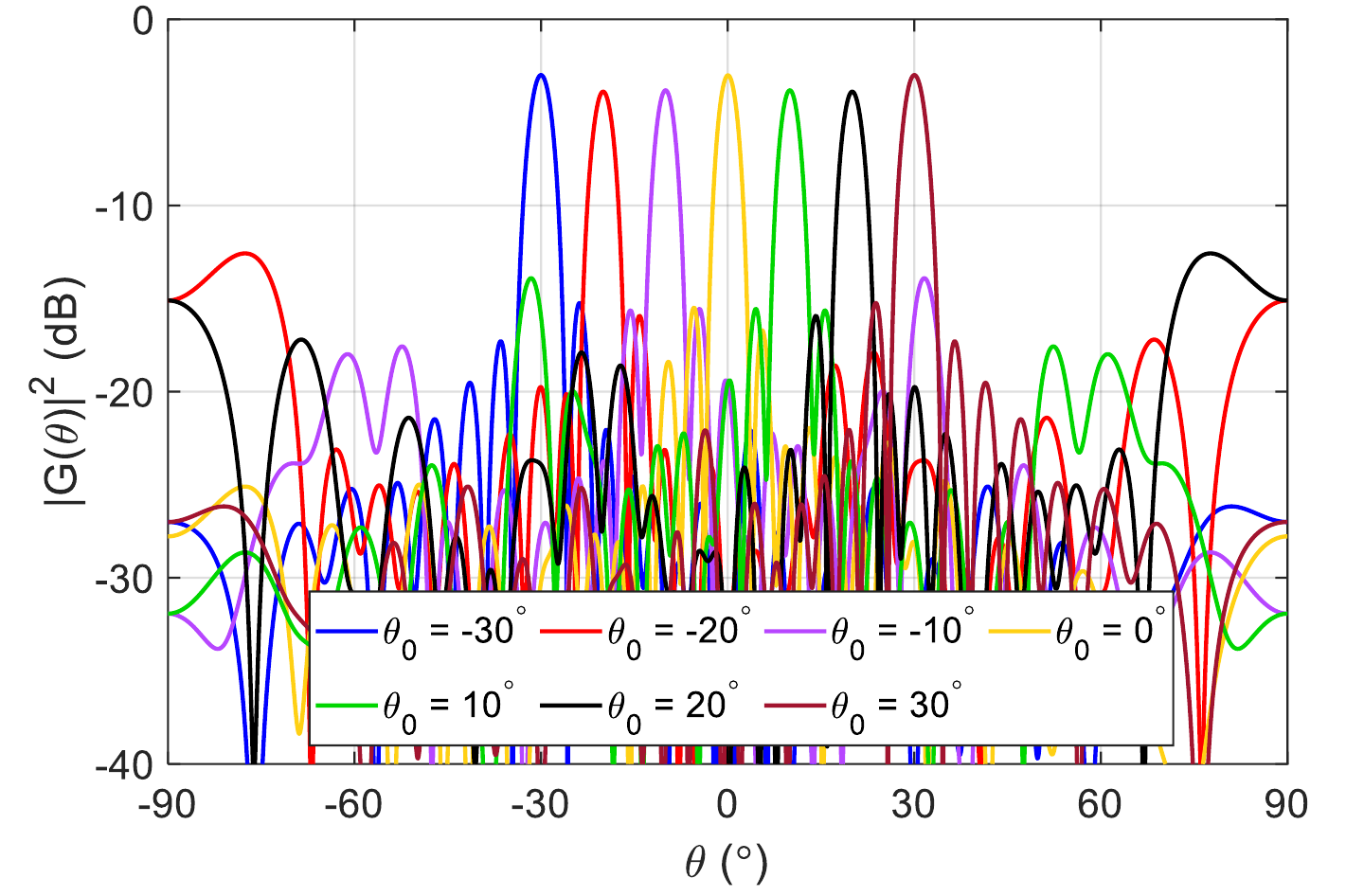}
    \caption{Radiation patterns along the $xz$ plane corresponding to $30\times 30$ IRS operating at normal incidence, for various $\theta_0$}
    \label{fig:prephasing_scanning}
\end{figure}

Finally, we comment on another outcome of the prephasing approach. While originally designed to eliminate the grating lobes issue at normal incidence, it turns out that it offers the opportunity to deploy the IRS at \emph{both} normal or non-normal incidence angle without any issue of grating lobes. We demonstrate this by applying the gOPA on a $30\times 30$ IRS, with the incident and required output wavevectors lying on the same vertical plane ($\phi_{in}=\phi_0 + 180^\circ$) and investigate the behavior w.r.t.~incidence angle after implementing prephasing with $\kappa=0.5$. The results are shown in Fig.~\ref{fig:prephase_heatmap}. We see that there are no grating lobes when operating at normal and non-normal incidence angle (with the exception for small regions at the corners), without having to tilt the IRS axis with respect to the ground.
 
\begin{figure}[htb!]
    \centering
    \includegraphics[width = 0.5 \textwidth]{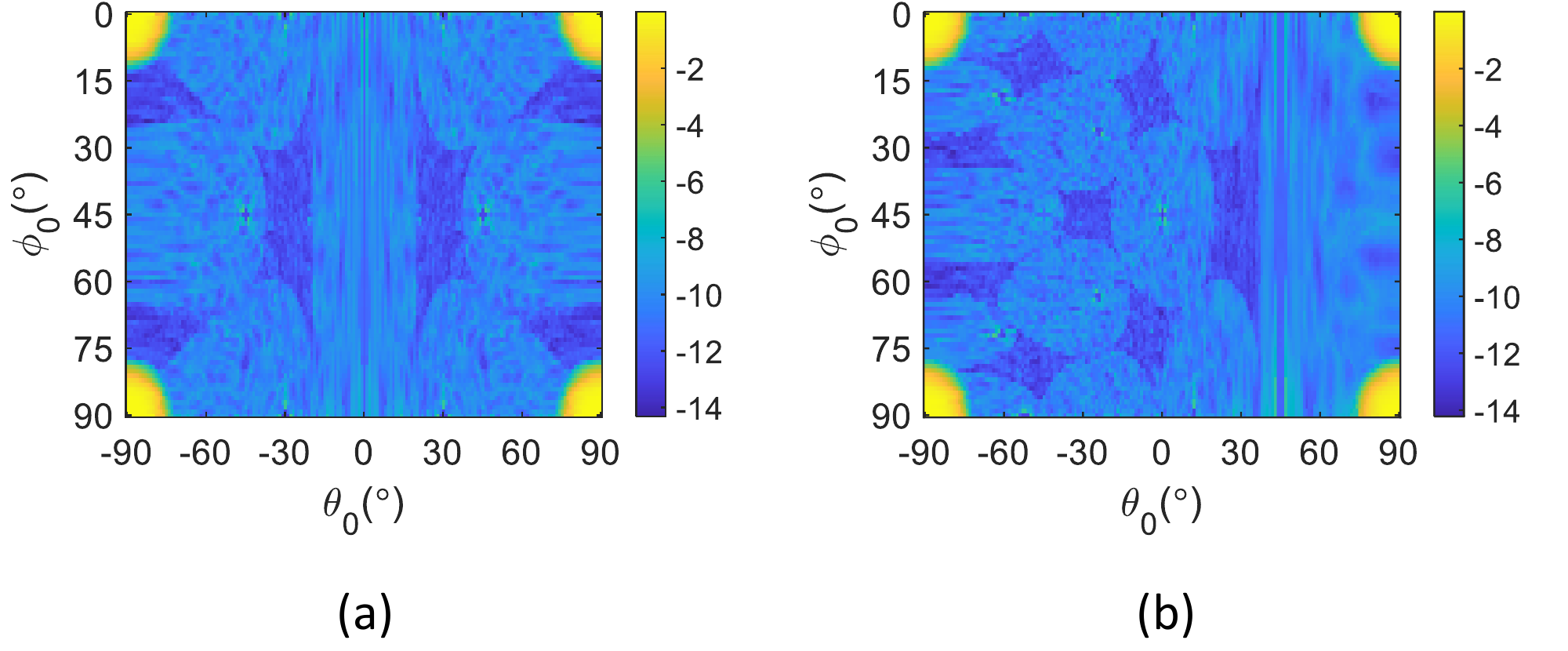}
    \caption{Sidelobe levels (in dB) on an IRS with prephasing on a $30\times 30$ rectangular IRS for different reflected beam angles ($\theta_0,\phi_0$), when incident and required wavevectors lie on the same vertical plane ($\phi_{in}=\phi_0 + 180^\circ$): (a) normal incidence, $\theta_{in}=0$, and (b) non-nomral incidence, $\theta_{in}=-45^\circ$.}
    \label{fig:prephase_heatmap}
\end{figure}

\section{Conclusion}

In this work, we have presented optimal beamforming algorithms in the context of intelligent reflecting surfaces realized with discrete phase shifts, both for the 1-bit case and the general case of $k$ phase shifts ($k > 2$). Our 1-bit algorithm can be adapted to allow each element of the IRS to select its phase from its own set of possible phases, thus being compatible with imperfect hardware realizations. They also outperform naive thresholding-based implementations. We also demonstrate beamforming in multiple directions. Further, our algorithms come with rigorous proofs of optimality as demonstrated in Theorems 1 and 2 -- a feature critically missing in heuristic algorithms. Due to the linear computational complexity of the proposed algorithms, they are extremely fast to deploy in practice, possibly even in real time applications. This makes them an attractive alternative to time-consuming evolutionary algorithms. 

We extensively study the issue of grating lobes and the various situations under which they arise via theory and numerical simulations. We discuss two alleviation strategies to mitigate these grating lobes. The first is the use of a triangular lattice in the IRS, and the second is the well known prephasing technique. We demonstrate the versatility of our algorithms by showing how beamforming can easily be accomplished in the presence of prephasing via extensive numerical studies. That said, further work is required to adapt them to more complex objective functions, such as those incorporating sidelobe level minimization, null beamforming directions, and beamforming over a certain bandwidth. We also assume planewave incidence on the IRS in our formulation, and the approach will require modification in the case of nearfield or spherical wave excitation. Our approach currently ignores the presence of mutual coupling between the reflecting elements, and incorporating this is an item of future research.

\appendices

\section{Proof of Theorem 1}
\label{theorem 1 proof}
Suppose that $\widetilde{w} = \begin{bmatrix}\widetilde{w}_1 & \dots & \widetilde{w}_n \end{bmatrix}$ is an optimal solution to the problem in \eqref{thm1}, and $\widetilde{s} = \sum_{i=1}^n \widetilde{w}_i z_i$. We claim that the line $\ell = \{ z \in \mathbb{C} \mid \operatorname{Re}(\widetilde{s} z^*) = 0 \}$ (\ie, the line perpendicular to $\widetilde{s}$ and passing through the origin) satisfies properties 1 and 2 stated in the theorem. We will require the following lemma.
\begin{lemma}
    Suppose $a$ and $b$ are non-zero complex numbers, such that $\re(a b^*) = \re(b a^*) \geq 0$. Then, $|a+b|^2 > |a|^2$.
\end{lemma}
\begin{proof}
    We have $|a+b|^2 = |a|^2 + |b|^2 + 2 \re(a b^*) > |a|^2$.
\end{proof}

To prove that $\ell$ satisfies property 1 (\ie~$\ell$ does not contain any $z_i$), assume for contradiction that it does not. Then, there exists some $j$ for which $z_j \in \ell$. Now, by switching $w_j = \widetilde{w}_j \rightarrow w_j = - \widetilde{w}_j $, we obtain the new sum $s = \sum_{i=1}^n w_i z_i = \widetilde{s} - 2 \widetilde{w}_j z_j$. Noting that $\operatorname{Re}(\widetilde{s} \cdot (2 \widetilde{w}_j z_j)^*) = 2 \widetilde{w}_j \operatorname{Re}(\widetilde{s} z_j^*) = 0$, by Lemma 1 we have $|s|^2 > |\widetilde{s}|^2$, which contradicts the optimality of $\widetilde{w}$. Therefore, property 1 has to be satisfied by $\ell$. 

To prove that $\ell$ satisfies property 2 (\ie~$\widetilde{w}_i = 1$ for all $z_i$ lying on one side of the line, and $\widetilde{w}_i = -1$ for all $z_i$ on the other side), we note that $\ell$ separates the complex plane into two half-planes given by: 
\begin{equation*}
\mathcal{H}_1 = \{ z \in \mathbb{C} \mid \operatorname{Re}(\widetilde{s} z^*) > 0 \}, \hspace{1.5mm}\mathcal{H}_2 = \{ z \in \mathbb{C} \mid \operatorname{Re}(\widetilde{s} z^*) < 0 \}.
\end{equation*}
We claim that $\widetilde{w}_j = 1$ for every $z_j \in \mathcal{H}_1$. To prove this, suppose the contrary that $\widetilde{w}_j = -1$ for some $z_j \in \mathcal{H}_1$. Then, by switching $w_j = \widetilde{w}_j = -1 \rightarrow w_j = 1$, we obtain the new sum $s = \sum_{i=1}^n w_i z_i = \widetilde{s} + 2 z_j$, and noting that $\operatorname{Re}(\widetilde{s} z_j^*) > 0$, we have $|s|^2  > |\widetilde{s}|^2$, contradicting the optimality of $\widetilde{w}_j$. A similar line of argument can be used to show that $\widetilde{w}_j = -1$ for every $z_j \in \mathcal{H}_2$. Hence, property 2 also holds.

\section{Proof of Theorem 2}
\label{theorem_2_proof}
Let $\widetilde{s} = \sum_{i=1}^n \widetilde{w}_i z_i$, where $\{\widetilde{w}_i\}$ represent the optimal solution to the problem in \eqref{thm2}, and define $\delta = \arg{ \widetilde{s}}$. We will show that this $\delta$ satisfies the required properties.

To prove that property 1 holds (\ie~none of the $z_i$ lie on the edges of the cones of the rotated partition  $\mathcal{P}(\delta)$), suppose for contradiction that some $z_i$ lies on the edge of two cones $A_j(\delta)$ and $A_l(\delta) \Longleftrightarrow e^{-j \delta} z_i$ lies on the edge of $A_j$ and $A_l$. Being a limit point of both $A_j$ and $A_l$, we have $\re( (a_j - a_m) z_i e^{-j \delta}) \geq 0$ $\forall$ $m \neq j$ (for $A_j(\delta)$), and similarly $\re( (a_l - a_m) z_i e^{-j \delta}) \geq 0$ $\forall$ $m \neq l$ (for $A_l(\delta)$). Together, this implies: $\operatorname{Re}( (a_j - a_l)e^{-j \delta} z_i) = 0$. We now have two cases:
\begin{enumerate}
    \item Suppose $\widetilde{w}_i \neq a_j$. Then, by switching from $w_i = \widetilde{w}_i$ to $w_i = a_j$, we obtain the new sum $s = \widetilde{s} + (a_j - \widetilde{w}_i) z_i$. Since $\re((a_j - \widetilde{w}_i) z_i \widetilde{s}^*) = |\widetilde{s}| \re( (a_j - \widetilde{w}_i) z_i e^{-j \delta}) \geq 0$, by Lemma 1 we have $|s|^2 > |\widetilde{s}|^2$, contradicting the optimality of $\widetilde{w}$.

    \item Suppose $\widetilde{w}_i = a_j$. Then, by switching from $w_i = a_j$ to $w_i = a_l$, we obtain the new sum $s = \widetilde{s} + (a_l - a_j) z_i$. Since $\re((a_l - a_j) z_i \widetilde{s}^*) = |\widetilde{s}| \re( (a_l - a_j) z_i e^{-j \delta}) = 0$, by Lemma 1 we have $|s|^2 > |\widetilde{s}|^2$, once again contradicting the optimality of $\widetilde{w}$.
    
\end{enumerate}
To prove property 2 holds (\ie~if  $z_i$  belongs to the cone $A_j(\delta)$ in the rotated partition $\mathcal{P}(\delta)$, then     $\widetilde{w}_i = a_{j}$), let $z_i \in A_j(\delta) \Longleftrightarrow e^{-j \delta} z_i \in A_j$ and suppose for contradiction that $\widetilde{w}_i \neq a_j$. Then, by switching from $w_i = \widetilde{w}_i$ to $w_i = a_j$, we obtain the new sum $s = \widetilde{s} + (a_j - \widetilde{w}_i) z_i$. Since $\re((a_j - \widetilde{w}_i) z_i \widetilde{s}^*) = |\widetilde{s}| \re( (a_j - \widetilde{w}_i) z_i e^{-j \delta}) > 0$ from the definition of $A_j$, by Lemma 1 we have $|s|^2 > |\widetilde{s}|^2$, contradicting the optimality of $\widetilde{w}$. Thus, $\widetilde{w}_i = a_j$, and hence property 2 holds.

\section{Derivation of Grating lobe condition for triangular lattice}
\label{triangular_app}
Define the matrix $D = \begin{bmatrix}
    d_1 & d_2
\end{bmatrix}$.
Then, we can write $x_{m,n} = D \begin{bmatrix}
    m & n
\end{bmatrix}^T$ for all $m,n$. 
Noting that $\varphi_{m,n}(\theta, \phi) = \frac{2 \pi}{\lambda}x_{m,n}^T \big(-\hat{u}(\theta, \phi) + \hat{u}(\pi - \theta_{in}, \phi_{in})\big)$, the grating lobe condition \eqref{gl_cond} becomes
\begin{multline}
   \frac{2 \pi}{\lambda}\begin{bmatrix}
        m & n
    \end{bmatrix} D^T \big( -\hat{u}(\theta^*, \phi^*) - \hat{u}(\theta_0, \phi_0) + 2 \hat{u}(\pi - \theta_{in}, \phi_{in}) \big) \\
    = \pi \begin{bmatrix}
        m & n
    \end{bmatrix}\begin{bmatrix}
        a & b
    \end{bmatrix}^T. 
\end{multline}
Since the above condition holds for all $m,n$, we must necessarily have:
\begin{align}
    \label{gl_eqn1_appendix}
    &\frac{2 \pi}{\lambda} D^T \big(- \hat{u}(\theta^*, \phi^*) - \hat{u}(\theta_0, \phi_0) + 2 \hat{u}(\pi - \theta_{in}, \phi_{in}) \big)
    = \pi \begin{bmatrix}
        a & b
    \end{bmatrix}^T.
\end{align}
Define $v = \begin{bmatrix}
      -\sin{\theta_0} \cos{\phi_0} - \sin{\theta^*} \cos{\phi^*} + 2 \sin{\theta_{in}} \cos{\phi_{in}} \\
      -\sin{\theta_0} \sin{\phi_0} - \sin{\theta^*} \sin{\phi^*} + 2 \sin{\theta_{in}} \sin{\phi_{in}}
 \end{bmatrix}$. \\
 Then, \eqref{gl_eqn1_appendix} becomes
\begin{align}   
    \label{gl_eqn2_appendix}
      &\frac{2 \pi d}{\lambda} \begin{bmatrix}
         1 & 0 \\ \frac{1}{2} & \frac{\sqrt{3}}{2}
     \end{bmatrix} v = \pi \begin{bmatrix}
         a \\ b
     \end{bmatrix} 
     \implies v = \frac{\lambda}{2d} \begin{bmatrix}
         a \\ \frac{-a+2b}{\sqrt{3}}
     \end{bmatrix}.
 \end{align}
Comparing the components on both sides of \eqref{gl_eqn2_appendix} gives us the desired \eqref{grating_lobe_condition_for_triangular_lattice}.
\subsection*{Analysis for $\phi_0 = 0^\circ, \phi_{in} = 180^\circ$}
The grating lobe condition now reduces to the following system of equations
\begin{subequations}
    \begin{align}
    2 \sin{\theta_{in}} + \sin{\theta_0} + a &= -\sin{\theta^*} \cos{\phi^*} \label{gl_eqn3a_appendix}, \\
    \frac{-a + 2b}{\sqrt{3}} &= -\sin{\theta^*} \sin{\phi^*}. \label{gl_eqn3b_appendix}
\end{align}
\end{subequations}
Since $\sin{\theta^*} \sin{\phi^*} \in [-1, 1]$ and $-a+2b$ is even, \eqref{gl_eqn3b_appendix} forces $a = 2b$. Since $b$ is even, $a$ is a multiple of $4$. However, \eqref{gl_eqn3a_appendix} forces $a \in \{-2,0,2 \}$, because $\sin{\theta^*} \cos{\phi^*} \in [-1, 1]$ and $2 \sin{\theta_{in}} + \sin{\theta_0} \in (-3, 3) $. Therefore, we must have $a = b = 0$, and hence \eqref{gl_eqn3a_appendix} and \eqref{gl_eqn3b_appendix} reduce to $2 \sin{\theta_{in}} + \sin{\theta_0} = -\sin{\theta^*} \cos{\phi^*}$ and $0 = \sin{\theta^*} \sin{\phi^*}$ respectively.

Now, if $- 2 \sin{\theta_{in}} - \sin{\theta_0} > 1 $ (equivalently, $\theta_0 < \arcsin{( - 2 \sin{\theta_{in}} - 1)} $), then clearly \eqref{gl_eqn3a_appendix} cannot be satisfied by any $(\theta^*, \phi^*)$, and hence a grating lobe will not exist. In particular, for $\theta_{in} = -45^\circ$, a grating lobe will not exist when $\theta_0 < \arcsin{(\sqrt{2}-1)} \approx 24.5^\circ$, which is observed in Fig.~\ref{triangular_theo_sim}.

\section{Proof of theorem 3}
\label{theorem_3_proof}
We recall the following relations mentioned in the theorem:
\begin{itemize}
       \item $w^{(k)} = \arg \max_{w_i} |G_1 + \sum_{j=2}^l \alpha_j^{(k)} G_j|$.
       \item $\alpha_j^{(k+1)} = e^{j \big( \arg(G_1) - \arg(G_j) \big)}$ for $j = 2, \hdots, l$, where each $G_j$ is computed using the weights $w^{(k)}$.
   \end{itemize}
We defined $c_k = \max_{w_i} |G_1 + \sum_{j=2}^l \alpha_j^{(k)} G_j|$, and $d_k = |G_1(w^{(k)} )| + \hdots + |G_l(w^{(k)})|$. Our goal is to show that the sequences $\{c_k\}_{k \geq 1}$ and $\{d_k\}_{k \geq 1}$ are both non-decreasing.

Now, we have $c_k = \max_{w_i} |G_1 + \sum_{j=2}^l \alpha_j^{(k)} G_j| = |G_1(w^{(k)}) + \sum_{j=2}^l \alpha_j^{(k)} G_j(w^{(k)}) |$ by definition of $w^{(k)}$. It follows from the triangle inequality that $c_k \leq d_k$.

Next, by definition of $\alpha_j^{(k+1)}$, we have $\arg{\big( G_1(w^{(k)}) \big)} = \arg{\big( \alpha_j^{(k+1)} G_j(w^{(k)}) \big)}$ for all $j \in \{2, \hdots, l\}$, and hence we have
\begin{align}
    \label{th4_eq1}
    |G_1(w^{(k)}) + \sum_{j=2}^l \alpha_j^{(k+1)} G_j(w^{(k)}) | &= \sum_{j=1}^l |G_j(w^{(k)} )| = d_k.
\end{align}
However, by the definition of $c_{k+1}$, we have
\begin{align}
    c_{k+1} &\geq |G_1(w^{(k)}) + \sum_{j=2}^l \alpha_j^{(k+1)} G_j(w^{(k)}) |,
\end{align}
and hence, from \eqref{th4_eq1}~we have $c_{k+1} \geq d_k$. Once again, we have $c_{k+1} \leq d_{k+1}$ due to the triangle inequality. Therefore, we have the chain of inequalities $c_k \leq d_k \leq c_{k+1} \leq d_{k+1}$ for all $k \geq 1$, which in particular implies that both the sequences $\{c_k\}_{k \geq 1}$ and $\{d_k\}_{k \geq 1}$ are non-decreasing, as desired.
\section*{Acknowledgment}
 The authors acknowledge useful discussions on beamforming with Dr.~Hanumantha Rao, Director General of SAMEER, India, Prof.~S.~Aniruddhan, Electrical Engineering, IIT Madras, and Mr. Amitesh Kumar of SAMEER, India.
 
\bibliographystyle{ieeetr}
\bibliography{refs}

\begin{thebibliography}{10}

\bibitem{9771079}
W.~Mei, B.~Zheng, C.~You, and R.~Zhang, ``Intelligent {Reflecting}
  {Surface-Aided} {Wireless} {Networks}: {From} {Single-Reflection} to
  {Multireflection} {Design} and {Optimization},'' {\em Proceedings of the
  IEEE}, vol.~110, no.~9, pp.~1380--1400, 2022.

\bibitem{7510962}
X.~Tan, Z.~Sun, J.~M. Jornet, and D.~Pados, ``Increasing {Indoor} {Spectrum}
  {Sharing} {Capacity} {Using} {Smart} {Reflect}-{Array},'' in {\em 2016 IEEE
  International Conference on Communications (ICC)}, pp.~1--6, 2016.

\bibitem{elbir2023twenty}
A.~M. Elbir, K.~V. Mishra, S.~A. Vorobyov, and R.~W. Heath, ``Twenty-{Five}
  {Years} of {Advances} in {Beamforming}: From {Convex} and {Nonconvex}
  {Optimization} to {Learning} {Techniques},'' {\em IEEE Signal Processing
  Magazine}, vol.~40, no.~4, pp.~118--131, 2023.

\bibitem{5447076}
A.~B. Gershman, N.~D. Sidiropoulos, S.~Shahbazpanahi, M.~Bengtsson, and
  B.~Ottersten, ``Convex {Optimization}-{Based} {Beamforming},'' {\em IEEE
  Signal Processing Magazine}, vol.~27, no.~3, pp.~62--75, 2010.

\bibitem{1420809_boyd_mvb}
R.~Lorenz and S.~Boyd, ``Robust {Minimum} {Variance} {Beamforming},'' {\em IEEE
  Transactions on Signal Processing}, vol.~53, no.~5, pp.~1684--1696, 2005.

\bibitem{8721535_vorobyov_mvdr}
Y.~Huang, M.~Zhou, and S.~A. Vorobyov, ``New {Designs} on {MVDR} {Robust}
  {Adaptive} {Beamforming} {Based} on {Optimal} {Steering} {Vector}
  {Estimation},'' {\em IEEE Transactions on Signal Processing}, vol.~67,
  no.~14, pp.~3624--3638, 2019.

\bibitem{1634819}
N.~Sidiropoulos, T.~Davidson, and Z.-Q. Luo, ``Transmit {Beamforming} for
  {Physical}-{Layer} {Multicasting},'' {\em IEEE Transactions on Signal
  Processing}, vol.~54, no.~6, pp.~2239--2251, 2006.

\bibitem{6952703}
{\"O}.~T. Demir and T.~E. Tuncer, ``Alternating {Maximization} {Algorithm} for
  the {Broadcast} {Beamforming},'' in {\em 2014 22nd European Signal Processing
  Conference (EUSIPCO)}, pp.~1915--1919, 2014.

\bibitem{8811733_irs_enhanced_joint_beamforming}
Q.~Wu and R.~Zhang, ``Intelligent {Reflecting} {Surface} {Enhanced} {Wireless}
  {Network} via {Joint} {Active} and {Passive} {Beamforming},'' {\em IEEE
  Transactions on Wireless Communications}, vol.~18, no.~11, pp.~5394--5409,
  2019.

\bibitem{9110889_hybrid_disc_beamforming}
B.~Di, H.~Zhang, L.~Song, Y.~Li, Z.~Han, and H.~V. Poor, ``Hybrid {Beamforming}
  for {Reconfigurable} {Intelligent} {Surface} based {Multi}-{User}
  {Communications}: {Achievable} {Rates} {With} {Limited} {Discrete} {Phase}
  {Shifts},'' {\em IEEE Journal on Selected Areas in Communications}, vol.~38,
  no.~8, pp.~1809--1822, 2020.

\bibitem{9133142_channel_estimation}
C.~You, B.~Zheng, and R.~Zhang, ``Channel {Estimation} and {Passive}
  {Beamforming} for {Intelligent} {Reflecting} {Surface}: {Discrete} {Phase}
  {Shift} and {Progressive} {Refinement},'' {\em IEEE Journal on Selected Areas
  in Communications}, vol.~38, no.~11, pp.~2604--2620, 2020.

\bibitem{9305278_iot_beamforming}
S.~Gong, Z.~Yang, C.~Xing, J.~An, and L.~Hanzo, ``Beamforming {Optimization}
  for {Intelligent} {Reflecting} {Surface}-{Aided} {SWIPT} {IoT} {Networks}
  {Relying} on {Discrete} {Phase} {Shifts},'' {\em IEEE Internet of Things
  Journal}, vol.~8, no.~10, pp.~8585--8602, 2021.

\bibitem{yin_single-beam_2020}
J.~Yin, Q.~Wu, Q.~Lou, H.~Wang, Z.~N. Chen, and W.~Hong, ``Single-{Beam} 1
  {Bit} {Reflective} {Metasurface} {Using} {Prephased} {Unit} {Cells} for
  {Normally} {Incident} {Plane} {Waves},'' {\em IEEE Transactions on Antennas
  and Propagation}, vol.~68, no.~7, pp.~5496--5504, 2020.

\bibitem{kashyap_mitigating_2020}
B.~G. Kashyap, P.~C. Theofanopoulos, Y.~Cui, and G.~C. Trichopoulos,
  ``Mitigating {Quantization} {Lobes} in {mmWave} {Low}-{Bit} {Reconfigurable}
  {Reflective} {Surfaces},'' {\em IEEE Open Journal of Antennas and
  Propagation}, vol.~1, pp.~604--614, 2020.

\bibitem{9807634_prephase_1-bit_low-sll}
J.~Yin, Q.~Wu, H.~Wang, and Z.~N. Chen, ``Prephase-{Based} {Equivalent}
  {Amplitude} {Tailoring} for {Low} {Sidelobe} {Levels} of 1-{Bit}
  {Phase}-{Only} {Control} {Metasurface} {Under} {Plane} {Wave} {Incidence},''
  {\em IEEE Transactions on Antennas and Propagation}, vol.~70, no.~11,
  pp.~10604--10613, 2022.

\bibitem{10195171_quantization_suppression}
D.~Vabichevich, A.~Belov, and A.~Sayanskiy, ``Suppression of {Quantization}
  {Lobes} in 1-{Bit} {Reconfigurable} {Intelligent} {Surfaces},'' {\em IEEE
  Antennas and Wireless Propagation Letters}, vol.~22, no.~12, pp.~2808--2811,
  2023.

\bibitem{10179250_prephase_low_res_hybrid_phasing}
M.~Cheng, Q.~Wu, C.~Yu, H.~Wang, and W.~Hong, ``A {Prephased} {Electronically}
  {Steered} {Phased} {Array} {That} {Uses} {Very}-{Low}-{Resolution} {Phase}
  {Shifters} and a {Hybrid} {Phasing} {Method},'' {\em IEEE Transactions on
  Antennas and Propagation}, vol.~71, no.~9, pp.~7310--7322, 2023.

\bibitem{smith1983comparison}
M.~Smith and Y.~Guo, ``A {Comparison} of {Methods} for {Randomizing} {Phase}
  {Quantization} {Errors} in {Phased} {Arrays},'' {\em IEEE transactions on
  antennas and propagation}, vol.~31, no.~6, pp.~821--828, 1983.

\bibitem{yang2017study}
H.~Yang, F.~Yang, S.~Xu, M.~Li, X.~Cao, J.~Gao, and Y.~Zheng, ``A {Study} of
  {Phase} {Quantization} {Effects} for {Reconfigurable} {Reflectarray}
  {Antennas},'' {\em IEEE Antennas and Wireless Propagation Letters}, vol.~16,
  pp.~302--305, 2017.

\bibitem{zhang_quadratic_2000}
S.~Zhang, ``Quadratic {Maximization} and {Semidefinite} {Relaxation},'' {\em
  Mathematical Programming}, vol.~87, no.~3, pp.~453--465, 2000.

\bibitem{ben-tal_lectures_2001}
A.~Ben-Tal and A.~Nemirovski, {\em Lectures on {Modern} {Convex}
  {Optimization}: {Analysis}, {Algorithms}, and {Engineering} {Applications}}.
\newblock Society for Industrial and Applied Mathematics, 2001.

\bibitem{4476079_no_of_bits}
B.~Wu, A.~Sutinjo, M.~E. Potter, and M.~Okoniewski, ``On the {Selection} of the
  {Number} of {Bits} to {Control} a {Dynamic} {Digital} {MEMS}
  {Reflectarray},'' {\em IEEE Antennas and Wireless Propagation Letters},
  vol.~7, pp.~183--186, 2008.

\bibitem{5068400_linear_reflectarray}
S.~Ebadi, R.~V. Gatti, and R.~Sorrentino, ``Linear {Reflectarray} {Antenna}
  {Design} {Using} 1-{Bit} {Digital} {Phase} {Shifters},'' in {\em 2009 3rd
  European Conference on Antennas and Propagation}, pp.~3729--3732, 2009.

\bibitem{7389996_hybrid_large_scale}
F.~Sohrabi and W.~Yu, ``Hybrid {Digital} and {Analog} {Beamforming} {Design}
  for {Large}-{Scale} {Antenna} {Arrays},'' {\em IEEE Journal of Selected
  Topics in Signal Processing}, vol.~10, no.~3, pp.~501--513, 2016.

\bibitem{zhang_reconfigurable_2020}
H.~Zhang, B.~Di, L.~Song, and Z.~Han, ``Reconfigurable {Intelligent} {Surfaces}
  {Assisted} {Communications} {With} {Limited} {Phase} {Shifts}: {How} {Many}
  {Phase} {Shifts} {Are} {Enough}?,'' {\em IEEE Transactions on Vehicular
  Technology}, vol.~69, pp.~4498--4502, Apr. 2020.

\bibitem{wu2019beamforming}
Q.~Wu and R.~Zhang, ``Beamforming optimization for wireless network aided by
  intelligent reflecting surface with discrete phase shifts,'' {\em IEEE
  Transactions on Communications}, vol.~68, no.~3, pp.~1838--1851, 2019.

\bibitem{yan2022machine}
W.~Yan, G.~Sun, W.~Hao, Z.~Zhu, Z.~Chu, and P.~Xiao, ``Machine learning-based
  beamforming design for millimeter wave irs communications with discrete phase
  shifters,'' {\em IEEE Wireless Communications Letters}, vol.~11, no.~12,
  pp.~2467--2471, 2022.

\bibitem{quantum2024}
Q.~J. Lim, C.~Ross, A.~Ghosh, F.~W. Vook, G.~Gradoni, and Z.~Peng,
  ``Quantum-{Assisted} {Combinatorial} {Optimization} for {Reconfigurable}
  {Intelligent} {Surfaces} in {Smart} {Electromagnetic} {Environments},'' {\em
  IEEE Transactions on Antennas and Propagation}, vol.~72, no.~1, pp.~147--159,
  2024.

\bibitem{pan2022overview}
C.~Pan, G.~Zhou, K.~Zhi, S.~Hong, T.~Wu, Y.~Pan, H.~Ren, M.~Di~Renzo, A.~L.
  Swindlehurst, R.~Zhang, {\em et~al.}, ``An overview of signal processing
  techniques for ris/irs-aided wireless systems,'' {\em IEEE Journal of
  Selected Topics in Signal Processing}, vol.~16, no.~5, pp.~883--917, 2022.

\bibitem{zhang2022configuring}
Y.~Zhang, K.~Shen, S.~Ren, X.~Li, X.~Chen, and Z.-Q. Luo, ``Configuring
  intelligent reflecting surface with performance guarantees: Optimal
  beamforming,'' {\em IEEE Journal of Selected Topics in Signal Processing},
  vol.~16, no.~5, pp.~967--979, 2022.

\bibitem{ren2022linear}
S.~Ren, K.~Shen, X.~Li, X.~Chen, and Z.-Q. Luo, ``A linear time algorithm for
  the optimal discrete irs beamforming,'' {\em IEEE Wireless Communications
  Letters}, vol.~12, no.~3, pp.~496--500, 2022.

\bibitem{beamforming_conf_2024}
S.~S. Narayanan, R.~K. Ganti, and U.~K. Khankhoje, ``Discrete {Variable}
  {Beamforming} for {Intelligent} {Reflecting} {Surfaces},'' in {\em 2024 IEEE
  International Symposium on Antennas and Propagation and USNC-URSI Radio
  Science Meeting (AP-S/URSI)}, 2024 [to appear].

\bibitem{elachi1990radar}
C.~Elachi, Y.~Kuga, K.~McDonald, K.~Sarabandi, F.~Ulaby, M.~Whitt, H.~Zebker,
  and J.~van Zyl, {\em Radar polarimetry for geoscience applications}.
\newblock Norwood, MA (USA); Artech House Inc., 1990.

\bibitem{balanis2016antenna}
C.~A. Balanis, {\em Antenna theory: analysis and design}.
\newblock John wiley \& sons, 2016.

\bibitem{bhattacharyya2006phased}
A.~K. Bhattacharyya, {\em Phased array antennas: Floquet analysis, synthesis,
  BFNs and active array systems}, vol.~179.
\newblock John Wiley \& Sons, 2006.

\end{thebibliography}

\begin{IEEEbiographynophoto}{Sai Sanjay Narayanan}
is currently an undergraduate student in the Engineering Physics program at the Indian Institute of Technology Madras, Chennai.
\end{IEEEbiographynophoto}

\begin{IEEEbiographynophoto}{Uday K Khankhoje}
(Senior Member, IEEE) received the B.Tech.~degree in Electrical Engineering from the Indian Institute of Technology Bombay, Mumbai, in 2005, and the M.S.~and Ph.D.~degrees in Electrical Engineering from the California
Institute of Technology, Pasadena, in 2010. He was a Caltech Postdoctoral Scholar at the Jet Propulsion Laboratory (NASA/Caltech) from 2011-12, a Postdoctoral Research Associate in the Department of Electrical Engineering at the University of Southern California, Los Angeles, USA, from 2012-13, and an Assistant Professor of Electrical Engineering at the Indian Institute of Technology Delhi, India from 2013-16. Since 2016 he has been at the Department of Electrical Engineering in the Indian Institute of Technology Madras, Chennai, where he is currently an Associate Professor and leads the numerical electromagnetics and optics (NEMO) research group which focuses on solving electromagnetics inspired inverse problems.
\end{IEEEbiographynophoto}

\begin{IEEEbiographynophoto}{Radha Krishna Ganti}
(Member, IEEE) received the B.Tech. and M.Tech. degrees in Electrical Engineering from the Indian Institute
of Technology Madras, Chennai, India in 2004, and the Master’s degree in Applied Mathematics and the Ph.D. degree in Electrical Engineering from the University of Notre Dame in 2009. He is currently a Professor of Electrical Engineering at the Indian Institute of Technology Madras, Chennai. He is a coauthor of the monograph,\textit{ Interference in Large Wireless Networks} (NOW Publishers, 2008). His doctoral work focused on the spatial analysis of interference networks using tools from stochastic geometry. He received the 2014 IEEE Stephen O. Rice Prize, the 2014 IEEE Leonard G. Abraham Prize, and the 2015 IEEE Communications Society Young Author Best Paper Award. He was also awarded the 2016–2017 Institute Research and Development Award (IRDA) by IIT Madras. In 2019, he was awarded the TSDSI Fellow for Technical Excellence in standardization activities and contribution to LMLC use case in ITU. He was the Lead PI from IITM involved in the development of 5G base stations for the 5G Testbed Project funded by the Department of Telecommunications, Government of India.
\end{IEEEbiographynophoto}
\end{document}